\documentclass[onecolumn,10pt,aps,prd,preprintnumbers,nofootinbib,amsmath,amssymb,floatfix,longbibliography, showkeys]{revtex4-2}

\usepackage{bm}
\usepackage{booktabs}
\usepackage{graphicx}
\usepackage{xcolor}
\usepackage{hyperref}
\usepackage{microtype}
\hypersetup{colorlinks=true,linkcolor=blue,citecolor=blue,urlcolor=blue}
\allowdisplaybreaks

\newtheorem{theorem}{Theorem}[section]

\newtheorem{corollary}[theorem]{Corollary}
\numberwithin{equation}{section}

\begin{document}

\title{Photon sphere jet equivalence and the non-uniqueness of black-hole thermodynamics}

\author{Nikko John Leo S. Lobos}
\email{nikko\_john\_s\_lobos@dlsu.edu.ph}
\affiliation{Department of Physics, De La Salle University, 2401 Taft Avenue, Malate, Manila 1004, Philippines}

\author{Emmanuel T. Rodulfo}
\email{emmanuel.rodulfo@dlsu.edu.ph}
\affiliation{Department of Physics, De La Salle University, 2401 Taft Avenue, Malate, Manila 1004, Philippines}

\begin{abstract}
  We examine whether scalar eikonal quasinormal-mode data determine black-hole thermodynamics. For static, spherically symmetric, asymptotically flat, nonextremal geometries, we prove that the formal fixed-overtone expansion of a neutral, minimally coupled, massless scalar through order $L^{-N}$ depends only on the $(2N+2)$ jet of the lapse at the unstable photon sphere. For every finite $N$, real-analytic deformations can preserve this jet and the asymptotic Reissner--Nordstr\"om coefficients while changing the surface gravity and Hawking temperature. We also construct a nonanalytic $C^\infty$ deformation that is flat at the photon sphere. This deformation preserves the complete formal eikonal series to every algebraic order but changes the horizon position, temperature, and Einstein area entropy. A nonempty charged subfamily satisfies the weak energy condition as an effective Einstein source. Together, these four results establish the new finite-jet and all-orders thermodynamic obstruction. To separate geometric information from action-dependent information, we incorporate a previously established Reissner--Nordstr\"om--MOG comparison. The mapped backgrounds have identical normalized metrics, neutral-scalar boundary-value problems, exact scalar spectra, and Hawking temperatures, yet their Wald entropies differ because their curvature couplings differ. The resulting hierarchy shows that local photon-sphere data do not determine horizon thermodynamics and that even exact action-blind scalar data do not determine gravitational entropy. Agreement of the formal eikonal expansions for the deformed metrics does not imply exact finite-$\ell$ isospectrality. Global differences may appear beyond every algebraic order. Thermodynamic reconstruction from eikonal ringdown is therefore conditional on assumptions about the global metric and the gravitational action.
\end{abstract}

\maketitle

\section{Introduction}
\label{sec_introduction}

Black-hole quasinormal modes are scattering resonances determined by a wave equation and boundary conditions at the horizon and infinity. In the fixed-overtone eikonal limit, however, the algebraic large-angular-momentum expansion is controlled by the potential near its unstable maximum. For test fields on static, spherically symmetric geometries, the leading terms depend on the angular frequency and Lyapunov exponent of the unstable circular null orbit \cite{CardosoEtAl2009}. Systematic inverse-$L$ expansions also retain higher derivatives of the same local barrier \cite{DolanOttewill2009}. This localization makes eikonal modes useful probes of the photon sphere, but it raises an inverse question distinct from parameter estimation within a prescribed black-hole family. Which geometric and thermodynamic data does the local photon-sphere information in a formal eikonal expansion actually determine?

We address this question for a neutral, minimally coupled, massless test scalar on a one-function Schwarzschild-gauge metric. We assume that the spacetime is static, spherically symmetric, asymptotically flat, and nonextremal. Let $M$ be the asymptotic geometric mass scale. We use the dimensionless radius $x=r/M$, the asymptotically normalized time $\bar t=t/M$, and the lapse $F(x)$. The angular variable is $L=\ell+1/2$, and the overtone number $n$ remains fixed as $L\to\infty$. Let $x_c$ be the photon-sphere radius, let $j_{x_c}^{J_N}F:=(F_c,F_c',\ldots,F_c^{(J_N)})$ be the lapse jet through order $J_N$, and let $\mathcal D_N[F]:=(c_{-1,n}[F],\ldots,c_{N,n}[F])$ be the formal scalar eikonal data through $L^{-N}$. The locality theorem below establishes the existence of a finite-dimensional function $\mathfrak d_N$ such that
\begin{equation}
  \mathcal D_N[F]=\mathfrak d_N\!\left(j_{x_c}^{J_N}F\right).
  \label{eq_eikonal-functional-introduction}
\end{equation}
If $x_h$ is the simple outer horizon, the asymptotically normalized Hawking temperature is instead the functional
\begin{equation}
  T_H[F]=\frac{F'(x_h)}{4\pi M}.
  \label{eq_temperature-functional-introduction}
\end{equation}
Finally, if $I_{\rm grav}$ denotes the gravitational action, the Wald entropy is an action-dependent functional of the metric,
\begin{equation}
  S_{\rm W}[F,I_{\rm grav}]=\mathfrak S_{I_{\rm grav}}[F].
  \label{eq_entropy-functional-introduction}
\end{equation}
Eq. \eqref{eq_eikonal-functional-introduction} probes local trapping data, Eq. \eqref{eq_temperature-functional-introduction} probes global horizon data \cite{Hawking1975}, and Eq. \eqref{eq_entropy-functional-introduction} also depends on the curvature couplings in the gravitational action \cite{Wald1993,IyerWald1994}. In Einstein gravity, the latter reduces to the area law.

Imposing a functional form before posing the inverse problem hides this distinction. For example, within the Reissner--Nordstr\"om family $F_u(x)=1-2/x+u/x^2$, the photon-sphere invariants depend on the single dimensionless parameter $u$. Once one assumes this ansatz and the overall mass scale, suitable eikonal coefficients can be inverted to recover the family parameters, after which the metric determines the temperature. This closed inverse remains conditional on the ansatz. A nonvanishing Jacobian or an inverse-function theorem establishes local injectivity of the selected parameter map. It does not show that the measured coefficients select the RN functional form from the space of general lapse functions. Dimensional scaling determines how frequencies and temperatures vary with $M$, but it supplies neither the missing global metric profile nor the curvature coefficient of an unknown gravitational action.

We do not claim these two observations as new. Instead, we replace the family-restricted numerical reconstruction with four theorems about local lapse jets and explicit deformation families in the stated regularity classes. We retain the RN--MOG comparison in Sec. \ref{sec_operator-action} as a previously established action-level ingredient needed to complete the identifiability hierarchy.

Several related inverse problems help define the scope of the present one. Parametrized scalar calculations have exhibited metric degeneracies in which non-Schwarzschild geometries approximate the leading Schwarzschild spectrum \cite{VolkelKokkotas2019}. Ringdown-based metric reconstruction is most sensitive near the circular photon orbit \cite{VolkelBarausse2020}, while theory-agnostic frameworks reconstruct combinations of perturbation potentials and couplings only after a parametrization has been chosen \cite{VolkelFranchiniBarausse2022}. A related light-ring analysis found that the potential value and its second derivative can be constrained more robustly than the individual coefficients of a generic potential deformation \cite{VolkelFranchiniBarausseBerti2022}. These studies demonstrate both the value of local barrier information and the degeneracies that arise in more detailed reconstructions.

Other studies relate photon-sphere or wave-scattering quantities to thermodynamics after fixing a black-hole family. Within charged RN and de Sitter RN families, Wei and Liu related eikonal null-orbit invariants to temperature extrema and thermodynamic phase structure \cite{WeiLiu2020ThermodynamicQNM}. V\"olkel, Konoplya, and Kokkotas instead studied an inverse problem based on Hawking-emission spectra and greybody factors, using those data to constrain barrier widths \cite{VolkelKonoplyaKokkotas2019Hawking}. More recently, a quantum-optical framework combined near-horizon detector thermality with photon-sphere QNM resonances in a prescribed Schwarzschild model \cite{Ovgun2026QuantumOptics}. These results describe model-restricted correlations or use different data channels. They do not establish that formal eikonal resonance data determine the horizon or the gravitational action.

The mathematical inverse-resonance literature provides a complementary local result. For an analytic semiclassical potential with a nondegenerate trapping point, Birkhoff-normal-form resonance data can determine Taylor coefficients under specified nonresonance and reflection-symmetry assumptions \cite{Iantchenko2008}. Quantum Birkhoff normal forms also yield precise barrier-top quasinormal asymptotics for de Sitter--Reissner--Nordstr\"om black holes, although that global setting differs from the asymptotically flat problem considered here \cite{Iantchenko2017}. Positive recovery results are also available within a fixed model class. Uhlmann and Wang proved local uniqueness with H\"older stability for recovering the mass of a Schwarzschild--de Sitter black hole from one quasinormal mode \cite{UhlmannWang2023}. These theorems do not conflict with our nonuniqueness result. They recover local Taylor data or parameters within a specified analytic family, not an unrestricted global lapse or a gravitational action.

The unresolved step is the passage from local resonance data to horizon thermodynamics without first imposing a global metric family. Existing light-ring and inverse-resonance results do not determine whether a common finite or complete photon-sphere jet can be embedded in regular black-hole exteriors with different surface gravities and horizon areas. We address this gap by constructing explicit analytic and smooth families rather than relying on parameter counting. The new contribution is the constructive finite-jet and all-orders thermodynamic obstruction, together with the hierarchy that follows from it. We do not present either light-ring locality or the RN--MOG operator/entropy comparison as a new result.

The main limitation is important. For the deformed metrics constructed below, equality refers only to the formal fixed-overtone expansion in algebraic powers of $L^{-1}$. It does not imply equality of the exact quasinormal frequencies at finite $\ell$. Exact modes solve a global connection problem, so differences outside the common trapping region may produce a splitting beyond every algebraic order \cite{BertiCardosoStarinets2009,DyatlovZworski2019}. We assert exact scalar isospectrality only for the separate RN--MOG comparison, in which the complete normalized scalar boundary-value problems are identical. We make no claim about charged or nonminimal probes, gravitational perturbations, rotating or extremal geometries, or detector-level inference.

Within this scope, we prove four new technical theorems, incorporate one previously established action-level example, and derive a new hierarchy. The photon-sphere jet-locality theorem shows that the formal scalar expansion through $L^{-N}$ depends only on the $(2N+2)$-jet of $F$ at the trapping point, with a generically sharp derivative count. For every finite $N$, the real-analytic finite-order nonuniqueness theorem then constructs asymptotically RN lapses that preserve this jet and the RN mass and charge coefficients while changing the Hawking temperature. The smooth all-orders thermodynamic nonuniqueness theorem provides a $C^\infty$ deformation that is flat at the photon sphere and agrees with RN on and outside that radius, yet moves the horizon and changes both $T_H$ and the Einstein area entropy. The effective-source admissibility theorem identifies a nonempty charged wedge of the same smooth family whose effective Einstein source satisfies the weak energy condition. The deformed horizon still exhibits a dominant-energy violation. The previously established scalar-operator/action comparison supplies RN and constant-coupling MOG descriptions with identical normalized metrics, neutral-scalar operators, exact scalar spectra, and temperatures but unequal Wald entropies. Finally, the three-level identifiability hierarchy separates finite formal data, complete formal data, and exact action-blind scalar data. It follows from these four theorems and the earlier action-level example. Throughout the paper, we use natural units $c=\hbar=k_{\mathrm B}=1$ and retain $G_N$ explicitly.

The paper follows the order of these results. Section~\ref{sec_jet-locality} proves photon-sphere jet locality and the bound $J_N=2N+2$. Section~\ref{sec_analytic-nonuniqueness} establishes real-analytic finite-order temperature nonuniqueness. Section~\ref{sec_smooth-nonuniqueness} constructs the smooth all-orders family with variable horizon thermodynamics, and Sec. \ref{sec_effective-source} proves a nonempty weak-energy-condition domain for that family. Section~\ref{sec_operator-action} rederives the previously established exact neutral-scalar operator equivalence and action-dependent entropy comparison. Section~\ref{sec_identifiability-hierarchy} then derives the identifiability hierarchy, states the limitations, and presents the conclusions.

\section{Photon-sphere jet locality of the scalar eikonal expansion}
\label{sec_jet-locality}

We work with the dimensionless static, spherically symmetric line element
\begin{equation}
  ds^2=M^2\left[-F(x)d\bar t^2+F(x)^{-1}dx^2+x^2d\Omega_2^2\right],
  \label{eq_metric-section2}
\end{equation}
where $x=r/M$, $\bar t=t/M$, and $F>0$ in the exterior region containing the photon sphere. A prime denotes differentiation with respect to $x$. We assume a simple horizon, although the result below uses the metric only near the exterior trapping point.

A neutral, minimally coupled, massless scalar field obeys $\Box\Phi=0$. Because $\sqrt{-g}=M^4x^2\sin\theta$, the wave operator for Eq. \eqref{eq_metric-section2} is
\begin{equation}
  M^2\Box\Phi=-\frac{1}{F}\frac{\partial^2\Phi}{\partial\bar t^2}
  +\frac{1}{x^2}\frac{\partial}{\partial x}\left(x^2F\frac{\partial\Phi}{\partial x}\right)
  +\frac{1}{x^2}\Delta_{S^2}\Phi .
  \label{eq_box-section2}
\end{equation}
With the dimensionless frequency $\bar\omega:=M\omega$, separate the field as
\begin{equation}
  \Phi=e^{-i\bar\omega\bar t}Y_{\ell m}(\theta,\phi)\frac{\psi(x)}{x}.
  \label{eq_scalar-separation}
\end{equation}
Here $Y_{\ell m}$ is a scalar spherical harmonic with angular indices $(\ell,m)$ and eigenvalue $\Delta_{S^2}Y_{\ell m}=-\ell(\ell+1)Y_{\ell m}$. Direct substitution yields
\begin{equation}
  F^2\psi''+FF'\psi'
  +\left\{\bar\omega^2-F\left[\frac{\ell(\ell+1)}{x^2}+\frac{F'}{x}\right]\right\}\psi=0.
  \label{eq_radial-x}
\end{equation}
Define the dimensionless tortoise coordinate by $dx_*/dx=F^{-1}$. Then $\partial_*=\partial/\partial x_*=F\partial_x$, so the first two terms in Eq. \eqref{eq_radial-x} combine into $\partial_*^2\psi$. The exact radial equation becomes
\begin{equation}
  \frac{d^2\psi}{dx_*^2}+\left[\bar\omega^2-V_\ell(x)\right]\psi=0.
  \label{eq_scalar-radial}
\end{equation}
Here the effective potential is
\begin{equation}
  V_\ell(x)=F(x)\left[\frac{\ell(\ell+1)}{x^2}+\frac{F'(x)}{x}\right].
  \label{eq_scalar-potential}
\end{equation}
Introduce $L=\ell+1/2$, the angular-momentum variable used in systematic inverse-$L$ expansions of black-hole quasinormal frequencies \cite{DolanOttewill2009}. The identity $\ell(\ell+1)=L^2-1/4$ then gives
\begin{equation}
  V_L(x)=L^2U(x)+W(x).
  \label{eq_potential-split}
\end{equation}
The $L$-independent functions in this decomposition are defined separately by
\begin{subequations}
\label{eq_potential-components}
\begin{align}
  U(x)&=\frac{F(x)}{x^2},
  \label{eq_U-definition}\\
  W(x)&=F(x)\left[\frac{F'(x)}{x}-\frac{1}{4x^2}\right].
  \label{eq_W-definition}
\end{align}
\end{subequations}
Thus $U$ controls the leading large-$L$ barrier. Its exterior critical point $x_c$ is assumed to be a nondegenerate maximum, so $U''(x_c)<0$, and its stationary-point equation is
\begin{equation}
  U'(x_c)=0
  \quad\Longleftrightarrow\quad
  x_cF'_c-2F_c=0.
  \label{eq_photon-condition}
\end{equation}
where $F_c=F(x_c)$. The radial null-geodesic potential $\mathcal L^2F/x^2$, with $\mathcal L$ the conserved angular momentum per unit affine normalization, gives the same condition and identifies $x_c$ as the radius of the unstable circular null orbit. This relation underlies the leading eikonal correspondence between quasinormal frequencies and unstable circular null geodesics for test fields on static, spherically symmetric black holes \cite{CardosoEtAl2009}. Because $F_c>0$, a maximum with respect to $x$ is also a maximum with respect to $x_*$. In particular,
\begin{equation}
  U_{*,2}:=\left.\partial_*^2U\right|_{x_c}=F_c^2U''(x_c)<0.
  \label{eq_tortoise-curvature}
\end{equation}

We obtain the leading eikonal coefficients directly from the barrier-top problem. Ref. \cite{SchutzWill1985} introduced barrier-top WKB methods for black-hole quasinormal modes, and Ref. \cite{IyerWill1987} developed them systematically to higher order. Here we derive the coefficients in a semiclassical form suited to counting metric derivatives. Set the semiclassical parameter to $h:=L^{-1}$ and the rescaled spectral energy to $E:=h^2\bar\omega^2$. Define the corresponding operator by
\begin{equation}
  P_h=-h^2\partial_*^2+U+h^2W.
  \label{eq_semiclassical-operator}
\end{equation}
Eq. \eqref{eq_scalar-radial} then becomes $P_h\psi=E\psi$. Let $y:=x_*-x_{*c}$ measure tortoise-coordinate distance from the trapping point. The formal local expansions are
\begin{subequations}
\label{eq_tortoise-taylor}
\begin{align}
  U(y)=U_0+\sum_{m=2}^{\infty}\frac{U_m}{m!}y^m,
  \label{eq_U-tortoise-taylor}\\
  W(y)=\sum_{m=0}^{\infty}\frac{W_m}{m!}y^m,
  \label{eq_W-tortoise-taylor}
\end{align}
\end{subequations}
where $U_m=(\partial_*^mU)_c$, $W_m=(\partial_*^mW)_c$, $U_0=U(x_c)$, and $U_1=0$. Define
\begin{equation}
  \alpha=\sqrt{-\frac{U_2}{2}}>0.
  \label{eq_alpha-definition}
\end{equation}
The expansions in Eq. \eqref{eq_tortoise-taylor} represent formal jets at $x_c$ and do not require an analytic lapse. At each fixed order, the calculation uses only the corresponding finite Taylor polynomials.

At quadratic order, introduce the dimensionless scaled local coordinate $q$ through the complex scaling
\begin{equation}
  y=e^{i\pi/4}\sqrt{\frac{h}{\alpha}}\,q.
  \label{eq_complex-scaling}
\end{equation}
The scaled quadratic operator is
\begin{equation}
  \frac{P_h-U_0}{h}=-i\alpha\left(-\frac{d^2}{dq^2}+q^2\right)+O(h^{1/2}).
  \label{eq_inverted-oscillator}
\end{equation}
Its $n$th eigenvalue is $-i\alpha(2n+1)$. The formal barrier-top energy and frequency therefore begin as
\begin{align}
  E_n(h)&=U_0-i(2n+1)\alpha h+O(h^2),
  \label{eq_leading-energy}\\
  \bar\omega_{Ln}&=L\sqrt{U_0}
  -i\left(n+\frac12\right)\frac{\alpha}{\sqrt{U_0}}+O(L^{-1}).
  \label{eq_leading-frequency}
\end{align}
It follows that the leading frequency coefficients are
\begin{subequations}
\label{eq_leading-coefficients}
\begin{align}
  c_{-1,n}=\Omega_c,
  \label{eq_leading-real-coefficient}\\
  c_{0,n}=-i\left(n+\frac12\right)\Lambda_c,
  \label{eq_leading-damping-coefficient}
\end{align}
\end{subequations}
where the photon-sphere angular frequency and Lyapunov exponent are defined by
\begin{subequations}
\label{eq_omega-lambda}
\begin{align}
  \Omega_c^2=U_0=\frac{F_c}{x_c^2},
  \label{eq_omega-c-definition}\\
  \Lambda_c^2=\frac{\alpha^2}{U_0}
  =-\frac{F_c^2U''(x_c)}{2U(x_c)}
  =\frac{F_c\left(2F_c-x_c^2F''_c\right)}{2x_c^2}.
  \label{eq_lambda-c-definition}
\end{align}
\end{subequations}
The last equality follows from Eq. \eqref{eq_photon-condition} and shows explicitly that even the damping coefficient requires the second metric derivative at the photon sphere.

We next construct the complete formal recursion and use it to count the required derivatives. Set $\delta:=h^{1/2}$ and denote the harmonic-oscillator basis by $\{|m\rangle\}_{m\geq0}$, with $(-d^2/dq^2+q^2)|m\rangle=(2m+1)|m\rangle$. Define the quadratic operator $H_0:=-i\alpha(-d^2/dq^2+q^2)$. Equations~\eqref{eq_tortoise-taylor} and \eqref{eq_complex-scaling} then give
\begin{equation}
  \frac{P_h-U_0}{h}=H_0+\sum_{s=1}^{\infty}\delta^sH_s.
  \label{eq_normal-form-series}
\end{equation}
with
\begin{equation}
  H_s=A_sq^{s+2}+\mathbf 1_{s\geq2}B_sq^{s-2},
  \label{eq_normal-form-Hs}
\end{equation}
where the coefficients multiplying the two monomials are
\begin{subequations}
\label{eq_normal-form-AB}
\begin{align}
  A_s=\frac{e^{i\pi(s+2)/4}}{(s+2)!\,\alpha^{(s+2)/2}}U_{s+2},
  \label{eq_normal-form-A}\\
  B_s=\frac{e^{i\pi(s-2)/4}}{(s-2)!\,\alpha^{(s-2)/2}}W_{s-2}.
  \label{eq_normal-form-B}
\end{align}
\end{subequations}
The indicator in Eq. \eqref{eq_normal-form-Hs} removes the $B_s$ term for $s<2$. This form of the one-dimensional barrier-top normal form makes its dependence on the Taylor data explicit. In analytic settings, semiclassical Birkhoff normal forms relate trapping-point resonances to Taylor data at a nondegenerate critical point \cite{Iantchenko2008}. Related methods have also produced complete resonance asymptotics in black-hole scattering problems \cite{Iantchenko2017}. We use the recursion below only as a formal local construction, so a finite truncation requires finite differentiability, not analyticity.

The following recursion generates the coefficients. Expand the eigenvalue and state separately as
\begin{subequations}
\label{eq_eigen-expansions}
\begin{align}
  \lambda_n(\delta)=\sum_{s=0}^{\infty}\lambda_{n,s}\delta^s,
  \label{eq_eigenvalue-expansion}\\
  |\Psi_n(\delta)\rangle=|n\rangle+\sum_{s=1}^{\infty}\delta^s|\psi_{n,s}\rangle,
  \label{eq_eigenstate-expansion}
\end{align}
\end{subequations}
where $[H_0+\sum_{s\geq1}\delta^sH_s]|\Psi_n\rangle=\lambda_n|\Psi_n\rangle$ and $E_n=U_0+h\lambda_n$. We impose the intermediate normalization $\langle n|\psi_{n,s}\rangle=0$ and define $Q_n:=1-|n\rangle\langle n|$. Since the unperturbed eigenvalue is $\lambda_{n,0}:=-i\alpha(2n+1)$, the reduced resolvent is
\begin{equation}
  R_n=Q_n(H_0-\lambda_{n,0})^{-1}Q_n.
  \label{eq_reduced-resolvent}
\end{equation}
These expressions use the algebraic dual basis. We do not take the Hermitian conjugate of the complex-scaled coefficients. Matching equal powers of $\delta$ gives, for every $s\geq1$,
\begin{align}
  \lambda_{n,s}&=\sum_{r=1}^{s}\langle n|H_r|\psi_{n,s-r}\rangle,
  \label{eq_lambda-recursion}\\
  |\psi_{n,s}\rangle&=-R_n\sum_{r=1}^{s}
  \left(H_r-\lambda_{n,r}\right)|\psi_{n,s-r}\rangle.
  \label{eq_state-recursion}
\end{align}
Together with $|\psi_{n,0}\rangle=|n\rangle$, Eqs.~\eqref{eq_lambda-recursion} and \eqref{eq_state-recursion} determine the normal form recursively.

The apparent half-integer powers cancel from the eigenvalue. Both monomials in $H_s$ have parity $(-1)^s$. Induction in Eqs.~\eqref{eq_lambda-recursion} and \eqref{eq_state-recursion} therefore assigns $|\psi_{n,s}\rangle$ parity $(-1)^{n+s}$ and gives
\begin{equation}
  \lambda_{n,2j+1}=0,\qquad j\geq0.
  \label{eq_odd-vanishing}
\end{equation}
The formal energy consequently has the integer-power expansion
\begin{equation}
  E_n(h)\sim U_0+\sum_{j=1}^{\infty}e_{j,n}h^j.
  \label{eq_energy-expansion}
\end{equation}
Here $e_{j,n}:=\lambda_{n,2j-2}$ is the coefficient of $h^j$.
The recursion gives the first two energy coefficients as
\begin{align}
  e_{1,n}&=-i(2n+1)\alpha,
  \label{eq_e1}\\
  e_{2,n}&=W_0-
  \frac{U_4}{16\alpha^2}\left[\left(n+\frac12\right)^2+\frac14\right]
  -\frac{U_3^2}{1152\alpha^4}
  \left[60\left(n+\frac12\right)^2+7\right].
  \label{eq_e2}
\end{align}
In deriving Eq. \eqref{eq_e2}, $q^4$ contributes $3(2n^2+2n+1)/4$ on the $n$th oscillator state, while $q^3$ connects $n$ only to $n\pm1$ and $n\pm3$. Eq. \eqref{eq_e2} therefore confirms explicitly that the first subleading real coefficient involves $U_3$, $U_4$, and $W_0$, but no higher local data.

A second algebraic recursion gives the frequency coefficients. Define $a_{j,n}:=c_{j-1,n}$ and write
\begin{equation}
  \bar\omega_{Ln}=h^{-1}\sum_{j=0}^{\infty}a_{j,n}h^j.
  \label{eq_frequency-a-series}
\end{equation}
Since $E=h^2\bar\omega^2$, one has $a_{0,n}=\sqrt{U_0}$ and
\begin{equation}
  a_{j,n}=\frac{1}{2a_{0,n}}
  \left(e_{j,n}-\sum_{r=1}^{j-1}a_{r,n}a_{j-r,n}\right),
  \qquad j\geq1.
  \label{eq_frequency-recursion}
\end{equation}
For example, with $\nu=n+1/2$,
\begin{equation}
  c_{1,n}=\frac{e_{2,n}}{2\Omega_c}
  +\frac{\nu^2\alpha^2}{2\Omega_c^3}.
  \label{eq_c1}
\end{equation}
This expression reproduces Eqs.~\eqref{eq_leading-coefficients} at the preceding orders and exhibits the fourth-order potential dependence predicted by the general count below.

\begin{theorem}[Photon-sphere jet locality]
  \label{thm:photon-jet-locality}
  Fix $n\in\mathbb N_0$ and $N\in\mathbb N_0$. Let $F$ be positive and of class $C^{2N+2}$ in a neighborhood of $x_c$, and suppose that $U=F/x^2$ has a nondegenerate maximum there. Then the formal fixed-overtone expansion
  \begin{equation}
    M\omega_{Ln}\sim Lc_{-1,n}+c_{0,n}+c_{1,n}L^{-1}+\cdots
    \label{eq_formal-qnm-series}
  \end{equation}
  through the term $c_{N,n}L^{-N}$ depends only on
  \begin{equation}
    j_{x_c}^{J_N}F=\left(F_c,F'_c,\ldots,F_c^{(J_N)}\right).
    \label{eq_JN}
  \end{equation}
  Here $J_N=2N+2$ is the highest required derivative order.
  Thus two lapse functions with the same $(2N+2)$-jet at the same nondegenerate photon sphere have identical formal scalar eikonal coefficients from $c_{-1,n}$ through $c_{N,n}$. The bound is generically sharp. If two $C^\infty$ lapse functions share the complete jet at $x_c$, their formal fixed-overtone eikonal expansions agree at every algebraic order in $L^{-1}$.
\end{theorem}

\textit{Proof.} The recursion gives the derivative count directly. The coefficient $e_{j,n}=\lambda_{n,2j-2}$ can involve only $H_s$ with $s\leq2j-2$. Equations~\eqref{eq_normal-form-Hs} and \eqref{eq_normal-form-AB} then show that $e_{j,n}$ depends on $U_m$ only for $m\leq2j$ and on $W_m$ only for $m\leq2j-4$. The latter set is empty for $j=1$. This statement follows by induction. At recursion order $s$, Eq. \eqref{eq_state-recursion} uses only $H_1,\ldots,H_s$ and coefficients fixed at lower orders, while Eq. \eqref{eq_lambda-recursion} introduces no operator outside the same set. Therefore, no derivative of $U$ above $s+2$ or of $W$ above $s-2$ can enter $\lambda_{n,s}$.

We now translate the tortoise-coordinate derivatives into metric derivatives. Since $\partial_*=F\partial_x$, the two jet maps are
\begin{subequations}
\label{eq_tortoise-jet-map}
\begin{align}
  U_m=\left.(F\partial_x)^m\left(\frac{F}{x^2}\right)\right|_{x_c},
  \label{eq_U-tortoise-jet-map}\\
  W_m=\left.(F\partial_x)^m
  \left[F\left(\frac{F'}{x}-\frac{1}{4x^2}\right)\right]\right|_{x_c}.
  \label{eq_W-tortoise-jet-map}
\end{align}
\end{subequations}
Induction on $m$ shows that $U_m$ depends on $F_c,\ldots,F_c^{(m)}$, whereas $W_m$ depends on $F_c,\ldots,F_c^{(m+1)}$. The energy coefficient $e_{N+1,n}$, the highest coefficient needed for $c_{N,n}$ in Eq. \eqref{eq_frequency-recursion}, therefore uses $U_m$ through $m=2N+2$ and $W_m$ through $m=2N-2$. These terms require metric derivatives through orders $2N+2$ and $2N-1$, respectively. The first requirement is stronger, so $J_N=2N+2$.

For $N\geq1$, the term $A_{2N}q^{2N+2}$ contributes linearly to $\lambda_{n,2N}$ through $A_{2N}\langle n|q^{2N+2}|n\rangle$. This oscillator moment is nonzero, and $U_{2N+2}$ contains $F_c^{2N+2}F_c^{(2N+2)}/x_c^2$ together with terms formed from lower metric derivatives. Thus the highest derivative cannot be removed for a generic lapse. For $N=0$, Eq. \eqref{eq_omega-lambda} already contains $F''_c$. The bound is therefore generically sharp.

If two metrics have equal $(2N+2)$-jets, Eq. \eqref{eq_tortoise-jet-map} gives equal normal-form inputs through the required orders. The two recursions then yield equal $e_{j,n}$ for $j\leq N+1$ and equal $c_{k,n}$ for $-1\leq k\leq N$. Applying this finite-order statement for every $N$ proves the complete-jet corollary. A $C^{2N+2}$ lapse suffices for a fixed truncation, whereas a formal expansion to all algebraic orders requires $F\in C^\infty$ near $x_c$.

The theorem is local and asymptotic. It identifies the metric data that determine the formal barrier-top coefficients at fixed $n$ as $L\to\infty$. It does not imply that two metrics with matching photon-sphere jets have the same exact quasinormal frequencies at finite $\ell$. Exact modes obey global boundary conditions at the horizon and infinity \cite{BertiCardosoStarinets2009} and can retain information about the full exterior potential. Such information may enter beyond every algebraic order in $L^{-1}$ and is therefore not constrained by Theorem~\ref{thm:photon-jet-locality}.

\section{Real-analytic finite-order thermodynamic nonuniqueness}
\label{sec_analytic-nonuniqueness}

The locality theorem in Sec. \ref{sec_jet-locality} leaves open whether globally admissible black-hole exteriors with different horizon data can realize the same finite jets. We answer this question constructively. For a dimensionless charge parameter $u\in[0,1)$, define the horizon radii $x_\pm:=1\pm\sqrt{1-u}$. The reference RN lapse is
\begin{equation}
  F_u(x)=1-\frac{2}{x}+\frac{u}{x^2}
  =\frac{(x-x_+)(x-x_-)}{x^2}.
  \label{eq_RN-reference-section3}
\end{equation}
Its two circular null-orbit radii are the roots of $x^2-3x+2u=0$. We denote them by
\begin{subequations}
\label{eq_RN-photon-radii}
\begin{align}
  x_{\rm ph}^{-}&=\frac{3-\sqrt{9-8u}}{2},
  \label{eq_RN-inner-photon-radius}\\
  x_c&=\frac{3+\sqrt{9-8u}}{2}.
  \label{eq_RN-outer-photon-radius}
\end{align}
\end{subequations}
We keep the exterior root $x_c$ fixed throughout this section.

For an integer $p\geq2$ and a positive deformation parameter $\epsilon$, define
\begin{equation}
  F_{\epsilon,p}(x)=F_u(x)
  -\epsilon\frac{(x-x_+)(x-x_c)^{2p}}{x^{2p+4}}.
  \label{eq_analytic-deformation}
\end{equation}
The even power $2p$ gives the deformation a definite sign on the exterior, while allowing its multiplicity at $x_c$ to be arbitrarily large. To specify the exact admissible range, define
\begin{align}
  g_p(x)&=\frac{(x-x_-)x^{2p+2}}{(x-x_c)^{2p}},
  \label{eq_gp-definition}\\
  A_p&=(2p+3)x_c+2x_-,
  \label{eq_Ap-definition}\\
  \Delta_p=A_p^2-12(2p+2)x_-x_c,
  \label{eq_Ap-discriminant}\\
  \widehat{x}_p&=\frac{A_p+\sqrt{\Delta_p}}{6},
  \label{eq_xhat-definition}\\
  \epsilon_p^*(u)=
  \min\!\left\{g_p(x_+),g_p(\widehat{x}_p)\right\}.
  \label{eq_epsilon-star}
\end{align}
We use $\widehat{x}_p$ rather than $x_*$ because $x_*$ already denotes the tortoise coordinate in Sec. \ref{sec_jet-locality}.

\begin{theorem}[Real-analytic finite-order thermodynamic nonuniqueness]
  \label{thm:analytic-finite-nonuniqueness}
  Fix $u\in[0,1)$ and an eikonal truncation order $N\in\mathbb N_0$. Choose an integer
  \begin{equation}
    p\geq N+2,
    \label{eq_p-N-condition}
  \end{equation}
  and let $0<\epsilon<\epsilon_p^*(u)$. Then $F_{\epsilon,p}$ is real analytic on $x>0$, has the same asymptotic mass and RN charge coefficients as $F_u$, and shares its simple outer horizon $x_+$. Moreover, $F_{\epsilon,p}>0$ for every $x>x_+$, and $U_{\epsilon,p}=F_{\epsilon,p}/x^2$ has a unique global exterior maximum at $x_c$. For every fixed overtone $n$, the formal scalar eikonal coefficients satisfy
  \begin{equation}
    c_{k,n}[F_{\epsilon,p}]=c_{k,n}[F_u],
    \qquad -1\leq k\leq N.
    \label{eq_finite-eikonal-agreement}
  \end{equation}
  Nevertheless, the Hawking temperature differs from the RN value according to
  \begin{subequations}
  \label{eq_analytic-temperature-shift}
  \begin{align}
    T_H(\epsilon,p)=T_H^{\rm RN}
    -\frac{\epsilon}{4\pi M}
    \frac{(x_c-x_+)^{2p}}{x_+^{2p+4}},
    \label{eq_analytic-deformed-temperature}\\
    T_H^{\rm RN}=\frac{x_+-x_-}{4\pi Mx_+^2}.
    \label{eq_RN-temperature-section3}
  \end{align}
  \end{subequations}
  and therefore varies strictly with $\epsilon$. If we interpret the geometry in Einstein gravity, its horizon area and area entropy remain fixed as follows.
  \begin{subequations}
  \label{eq_fixed-area-entropy}
  \begin{align}
    A_H=4\pi M^2x_+^2,
    \label{eq_fixed-horizon-area}\\
    S_{\rm EH}=\frac{A_H}{4G_N}
    =\frac{\pi M^2x_+^2}{G_N}.
    \label{eq_fixed-Einstein-entropy}
  \end{align}
  \end{subequations}
  Thus a prescribed finite scalar eikonal expansion does not determine the temperature, even at fixed horizon area.
\end{theorem}

\textit{Proof.} We first establish the ordering and trapping properties of the reference geometry. Set $s=\sqrt{1-u}\in(0,1]$ and $d=\sqrt{9-8u}=\sqrt{1+8s^2}$. The identity
\begin{equation}
  d^2-(2s-1)^2=4s(s+1)>0
  \label{eq_ordering-identity}
\end{equation}
implies $|2s-1|<d$, and hence
\begin{equation}
  x_{\rm ph}^{-}<x_+<x_c.
  \label{eq_horizon-photon-ordering}
\end{equation}
Furthermore,
\begin{equation}
  \frac{d}{dx}\left(\frac{F_u}{x^2}\right)
  =-\frac{2(x-x_c)(x-x_{\rm ph}^{-})}{x^5}.
  \label{eq_RN-U-derivative}
\end{equation}
Because $x_{\rm ph}^{-}<x_+$, the reference potential $U_u=F_u/x^2$ increases on $(x_+,x_c)$ and decreases on $(x_c,\infty)$. It is positive between its zero at $x_+$ and its limit at infinity. Therefore, $x_c$ is its unique nondegenerate global exterior maximum.

Define the analytic prefactor $H_p(x):=(x-x_+)/x^{2p+4}$. The difference between the two lapse functions is
\begin{equation}
  F_{\epsilon,p}-F_u=-\epsilon(x-x_c)^{2p}H_p(x).
  \label{eq_deformation-multiplicity}
\end{equation}
The function $H_p$ is analytic at $x_c$, with $H_p(x_c)=(x_c-x_+)/x_c^{2p+4}>0$. Therefore,
\begin{align}
  \left.\frac{d^j}{dx^j}
  (F_{\epsilon,p}-F_u)\right|_{x_c}&=0,
  &&0\leq j\leq2p-1,
  \label{eq_analytic-jet-vanishing}\\
  \left.\frac{d^{2p}}{dx^{2p}}
  (F_{\epsilon,p}-F_u)\right|_{x_c}
  &=-\epsilon(2p)!\frac{x_c-x_+}{x_c^{2p+4}}\neq0.
  \label{eq_first-unmatched-derivative}
\end{align}
These relations prove exact jet agreement through order $2p-1$ and identify the first unprotected derivative.

The asymptotic behavior follows directly from
\begin{equation}
  \frac{(x-x_+)(x-x_c)^{2p}}{x^{2p+4}}
  =\frac{1}{x^3}\left(1-\frac{x_+}{x}\right)
  \left(1-\frac{x_c}{x}\right)^{2p}.
  \label{eq_deformation-asymptotic-factor}
\end{equation}
Hence,
\begin{equation}
  F_{\epsilon,p}(x)=1-\frac{2}{x}+\frac{u}{x^2}
  -\frac{\epsilon}{x^3}
  +\frac{\epsilon(x_++2px_c)}{x^4}+O(x^{-5}).
  \label{eq_analytic-asymptotic-series}
\end{equation}
The deformation changes neither the $x^{-1}$ ADM mass coefficient nor the $x^{-2}$ RN charge coefficient. Its first contribution appears at order $x^{-3}$.

We next prove that the proposed range of $\epsilon$ is sufficient and accounts for all possible exterior minima. Exact factorization gives
\begin{equation}
  F_{\epsilon,p}(x)=\frac{x-x_+}{x^{2p+4}}
  \left[(x-x_-)x^{2p+2}
  -\epsilon(x-x_c)^{2p}\right].
  \label{eq_deformed-lapse-factorization}
\end{equation}
For $x>x_+$ and $x\neq x_c$, positivity is therefore equivalent to $\epsilon<g_p(x)$. The function $g_p$ is positive on both sides of $x_c$, and its ordered triple of limiting values satisfies
\begin{equation}
  \left(\lim_{x\to x_c^-}g_p(x),
  \lim_{x\to x_c^+}g_p(x),
  \lim_{x\to\infty}g_p(x)\right)
  =(+\infty,+\infty,+\infty).
  \label{eq_gp-limits}
\end{equation}
Introduce the quadratic polynomial $Q_p(x):=3x^2-A_px+(2p+2)x_-x_c$. Its logarithmic derivative reads
\begin{equation}
  \frac{g_p'(x)}{g_p(x)}
  =\frac{Q_p(x)}{x(x-x_-)(x-x_c)}.
  \label{eq_gp-log-derivative}
\end{equation}
The roots of $Q_p$ are
\begin{subequations}
\label{eq_Q-roots}
\begin{align}
  \check{x}_p&=\frac{A_p-\sqrt{\Delta_p}}{6},
  \label{eq_Q-lower-root}\\
  \widehat{x}_p&=\frac{A_p+\sqrt{\Delta_p}}{6}.
  \label{eq_Q-upper-root}
\end{align}
\end{subequations}
The following exact signs determine their positions.
\begin{align}
  Q_p(x_+)&=(x_+-x_c)(3x_+-2x_-)
  +2px_c(x_--x_+)<0,
  \label{eq_Q-at-horizon}\\
  Q_p(x_c)&=-2px_c(x_c-x_-)<0,
  \label{eq_Q-at-photon}\\
  Q_p(0)&=(2p+2)x_-x_c\geq0.
  \label{eq_Q-at-zero}
\end{align}
Because $Q_p$ is an upward-opening quadratic, these signs imply
\begin{equation}
  0\leq\check{x}_p<x_+<x_c<\widehat{x}_p
  \label{eq_stationary-root-ordering}
\end{equation}
and, in particular, $\Delta_p>0$. On $(x_+,x_c)$, both $Q_p$ and the denominator in Eq. \eqref{eq_gp-log-derivative} are negative, so $g_p$ is strictly increasing and has infimum $g_p(x_+)$. On $(x_c,\widehat{x}_p)$, the denominator is positive and $Q_p<0$, whereas both are positive for $x>\widehat{x}_p$. Thus $g_p$ has a unique right-hand minimum at $\widehat{x}_p$, and
\begin{equation}
  \inf_{x>x_+,\,x\neq x_c}g_p(x)
  =\min\!\left\{g_p(x_+),g_p(\widehat{x}_p)\right\}
  =\epsilon_p^*(u)>0.
  \label{eq_global-g-minimum}
\end{equation}
Therefore, $0<\epsilon<\epsilon_p^*$ guarantees $F_{\epsilon,p}>0$ throughout $x>x_+$ and excludes any additional exterior horizon.

At $x=x_+$, the lapse vanishes and its derivative has the exact form
\begin{equation}
  F_{\epsilon,p}'(x_+)
  =\frac{(x_c-x_+)^{2p}}{x_+^{2p+4}}
  \left[g_p(x_+)-\epsilon\right]>0.
  \label{eq_simple-horizon-bound}
\end{equation}
The common horizon is therefore simple. The same sign choice also controls the full trapping potential.
\begin{equation}
  U_{\epsilon,p}(x)=U_u(x)
  -\epsilon\frac{(x-x_+)(x-x_c)^{2p}}{x^{2p+6}}
  \leq U_u(x),
  \qquad x\geq x_+.
  \label{eq_potential-ordering}
\end{equation}
Equality in the open exterior holds only at $x_c$. Because $U_u$ has its unique global exterior maximum there and the second derivatives agree by Eq. \eqref{eq_analytic-jet-vanishing}, $x_c$ remains the unique nondegenerate global maximum of $U_{\epsilon,p}$.

For a simple Killing horizon with the asymptotic time normalization of Eq. \eqref{eq_metric-section2}, the Hawking relation gives $T_H=F'(x_h)/(4\pi M)$ \cite{Hawking1975}. Direct differentiation at $x_+$ yields Eq. \eqref{eq_analytic-temperature-shift}, or equivalently
\begin{equation}
  \frac{T_H(\epsilon,p)}{T_H^{\rm RN}}
  =1-\frac{\epsilon}{g_p(x_+)},
  \label{eq_temperature-ratio}
\end{equation}
which is positive throughout the admissible domain and decreases strictly with $\epsilon$. Since the horizon radius is unchanged, the Einstein-gravity area law \cite{Bekenstein1973,Hawking1975} gives Eq. \eqref{eq_fixed-area-entropy}.

Finally, Theorem~\ref{thm:photon-jet-locality} gives $J_N=2N+2$, while condition~\eqref{eq_p-N-condition} implies
\begin{equation}
  2p-1\geq2N+3>J_N,
  \label{eq_jet-count-section3}
\end{equation}
so Eq. \eqref{eq_analytic-jet-vanishing} supplies every metric derivative required for Eq. \eqref{eq_finite-eikonal-agreement}. This completes the proof.

Figures~\ref{fig:analytic-lapse} and \ref{fig:analytic-potential} show the two nontrivial closed profiles for $u=1/2$ and $p=8$. Because $2p-1=15>J_6=14$, this choice preserves the formal eikonal expansion through at least $N=6$. It also shows why Eq. \eqref{eq_epsilon-star} requires both entries. For these parameters, the stationary value $g_p(\widehat{x}_p)$ is smaller than the horizon value $g_p(x_+)$. The plots illustrate the exact inequalities proved above but play no role in establishing positivity or jet agreement.

\begin{figure}[t]
  \centering
  \includegraphics[width=0.96\linewidth]{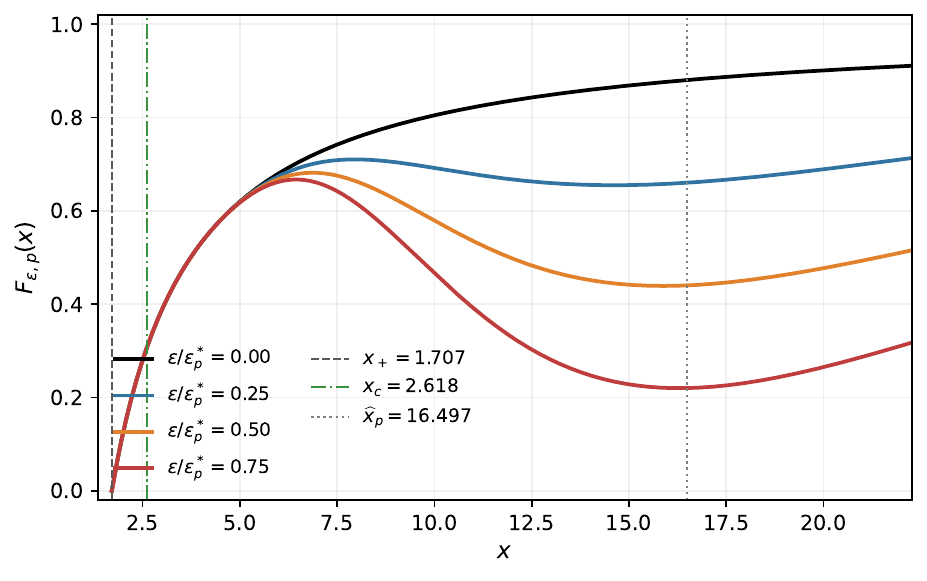}
  \caption{The lapse $F_{\epsilon,p}$ for $u=1/2$, $p=8$, and $\epsilon/\epsilon_p^*=0,0.25,0.50,0.75$. All curves have the same simple horizon $x_+$ and agree at the photon sphere $x_c$. The vertical dotted line marks $\widehat{x}_p$, where the right-hand positivity bound is attained.}
  \label{fig:analytic-lapse}
\end{figure}

\begin{figure}[t]
  \centering
  \includegraphics[width=0.74\linewidth]{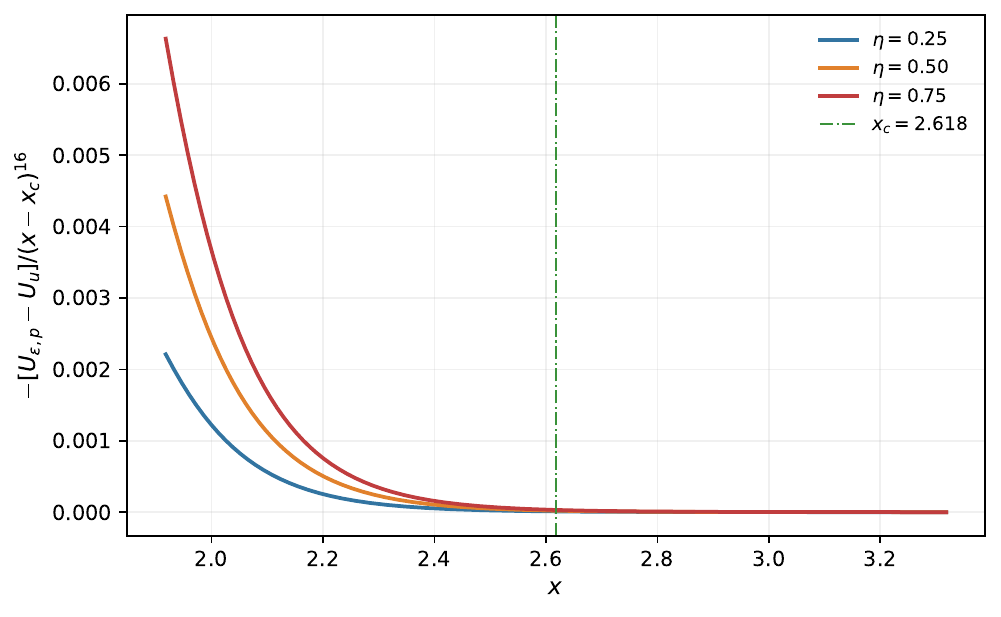}
  \caption{Regularized potential difference for the same family. The curves evaluate the analytic continuation of $-[U_{\epsilon,p}-U_u]/(x-x_c)^{16}=\epsilon(x-x_+)/x^{22}$, which shows explicitly that the first possible local departure is of order $16$. Equality through the fifteenth derivative follows exactly from Eq. \eqref{eq_analytic-jet-vanishing}.}
  \label{fig:analytic-potential}
\end{figure}

For $\eta=\epsilon/\epsilon_p^*$, the exact temperature law is $T_H/T_H^{\rm RN}=1-\eta\epsilon_p^*/g_p(x_+)$. It is affine on $0\leq\eta<1$, and the minimum in Eq. \eqref{eq_epsilon-star} fixes its endpoint analytically. A graph would only repeat Eq. \eqref{eq_temperature-ratio}, so we do not include one.

The theorem concerns the formal fixed-overtone expansion only through the chosen order $L^{-N}$. Because the deformation changes the potential away from $x_c$, we do not assert equality of the exact finite-$\ell$ quasinormal frequencies. For the minimal choice $p=N+2$, Eq. \eqref{eq_first-unmatched-derivative} places the first unprotected metric derivative precisely at the order that can enter the next coefficient $c_{N+1,n}$. The theorem is also geometric. It does not impose field equations or an energy condition on the effective source. Section~\ref{sec_effective-source} treats that separate admissibility question. Finally, a nonzero real-analytic deformation cannot be flat to every order at an interior point. The next section therefore uses a smooth, nonanalytic construction to obtain agreement at all algebraic eikonal orders while changing the horizon area.

\section{Smooth all-orders thermodynamic nonuniqueness}
\label{sec_smooth-nonuniqueness}

The analytic family in Sec. \ref{sec_analytic-nonuniqueness} preserves any prescribed finite photon-sphere jet, but it can preserve the complete jet only if the deformation vanishes identically. We therefore change the regularity class from real analytic to $C^\infty$. This change permits a deformation that is flat at the photon sphere but nonzero between the photon sphere and the horizon. Throughout this section, $F_u$, $x_+$, and $x_c$ denote the RN quantities defined in Eqs.~\eqref{eq_RN-reference-section3} and \eqref{eq_RN-photon-radii}.

Choose the parameter
\begin{equation}
  x_+<a<x_c
  \label{eq_smooth-a-domain}
\end{equation}
and, on the exterior $x\geq a$, define the cutoff $B_a$ and the deformed lapse $F_a$ by
\begin{subequations}
\label{eq_smooth-family}
\begin{align}
  B_a(x)=
  \begin{cases}
    \displaystyle
    \exp\!\left[-\frac{x-a}{x_c-x}\right],&a\leq x<x_c,\\[6pt]
    0,&x\geq x_c,
  \end{cases}
  \label{eq_smooth-cutoff}\\
  F_a(x)=F_u(x)-F_u(a)B_a(x).
  \label{eq_smooth-deformed-lapse}
\end{align}
\end{subequations}
The lower branch of $B_a$ is analytic near $x=a$ and can be continued across $a$. Thus Eq. \eqref{eq_smooth-family} defines an actual local zero at the proposed horizon, not merely a one-sided formal zero.

\begin{theorem}[Smooth all-orders thermodynamic nonuniqueness]
  \label{thm:smooth-all-orders}
  Fix $u\in[0,1)$ and let $a$ satisfy Eq. \eqref{eq_smooth-a-domain}. The lapse $F_a$ is $C^\infty$ on its exterior but not real analytic at $x_c$. It has a simple outer horizon at $x=a$, remains positive for every $x>a$, and equals $F_u$ exactly for $x\geq x_c$. The function $U_a=F_a/x^2$ has the same unique nondegenerate global exterior maximum $x_c$ as the RN reference. Moreover,
  \begin{equation}
    F_a^{(j)}(x_c)=F_u^{(j)}(x_c),
    \qquad j\in\mathbb N_0,
    \label{eq_complete-jet-agreement}
  \end{equation}
  and therefore, at every fixed overtone $n$,
  \begin{equation}
    c_{k,n}[F_a]=c_{k,n}[F_u],
    \qquad k=-1,0,1,\ldots .
    \label{eq_all-orders-eikonal-agreement}
  \end{equation}
  Eq. \eqref{eq_all-orders-eikonal-agreement} refers to the complete formal expansion in algebraic powers of $L^{-1}$.

  The horizon temperature and the Einstein-gravity area entropy are
  \begin{align}
    T_H(a)&=\frac{1}{4\pi M}
    \left[\frac{2(a-u)}{a^3}+\frac{F_u(a)}{x_c-a}\right],
    \label{eq_smooth-temperature}\\
    S_{\rm EH}(a)&=\frac{\pi M^2a^2}{G_N}.
    \label{eq_smooth-entropy}
  \end{align}
  The common formal eikonal series does not fix either quantity. In particular, the ordered pair of endpoint temperatures satisfies
  \begin{equation}
    \left(\lim_{a\to x_+^+}T_H(a),
    \lim_{a\to x_c^-}T_H(a)\right)
    =\left(T_H^{\rm RN},+\infty\right),
    \label{eq_smooth-temperature-limits}
  \end{equation}
  while the RN entropy limit and the entropy derivative form the ordered pair
  \begin{equation}
    \left(\lim_{a\to x_+^+}S_{\rm EH}(a),
    \frac{dS_{\rm EH}}{da}\right)
    =\left(S_{\rm RN},\frac{2\pi M^2a}{G_N}\right),
    \label{eq_smooth-entropy-variation}
  \end{equation}
  where $S_{\rm RN}:=\pi M^2x_+^2/G_N$. Because $a>0$, the derivative in Eq. \eqref{eq_smooth-entropy-variation} is strictly positive. Thus, in the $C^\infty$ metric class, the complete formal photon-sphere resonance germ does not determine the horizon position, temperature, or Einstein area entropy.
\end{theorem}

\textit{Proof.} At the proposed horizon, the exponent in Eq. \eqref{eq_smooth-family} vanishes, so $B_a(a)=1$. Therefore,
\begin{equation}
  F_a(a)=F_u(a)-F_u(a)B_a(a)=0.
  \label{eq_smooth-horizon-zero}
\end{equation}
For $a<x<x_c$, the exponent is strictly negative, so $0<B_a(x)<1$. The RN lapse is strictly increasing outside $x_+$ because
\begin{equation}
  F_u'(x)=\frac{2(x-u)}{x^3}>0,
  \qquad x>x_+,
  \label{eq_RN-monotonicity-section4}
\end{equation}
where $x_+-u=\sqrt{1-u}\,[1+\sqrt{1-u}]>0$. Hence,
\begin{equation}
  F_a(x)=F_u(x)-F_u(a)B_a(x)
  >F_u(x)-F_u(a)>0,
  \qquad a<x<x_c.
  \label{eq_smooth-positivity-inner}
\end{equation}
For $x\geq x_c$, $B_a=0$ and $F_a=F_u>0$. Therefore, $a$ is the only exterior horizon.

Differentiating the lower branch gives
\begin{equation}
  B_a'(x)=-\frac{x_c-a}{(x_c-x)^2}B_a(x).
  \label{eq_smooth-bump-first-derivative}
\end{equation}
At the horizon this derivative becomes $B_a'(a)=-1/(x_c-a)$.
It follows that
\begin{equation}
  F_a'(a)=F_u'(a)+\frac{F_u(a)}{x_c-a}
  =\frac{2(a-u)}{a^3}+\frac{F_u(a)}{x_c-a}>0.
  \label{eq_smooth-simple-horizon}
\end{equation}
The zero at $a$ is therefore simple and defines a nonextremal Killing horizon.

We next prove infinite-order matching at $x_c$. Let $d:=x_c-a>0$ denote the photon-sphere--horizon separation and let $t:=x_c-x$ denote the distance to the photon sphere from the left.
To the left of $x_c$, Eq. \eqref{eq_smooth-family} becomes
\begin{equation}
  B_a(x)=e^{1-d/t},
  \qquad t>0.
  \label{eq_standard-flat-form}
\end{equation}
Repeated differentiation shows by induction that, for every $j\in\mathbb N_0$,
\begin{equation}
  \frac{d^jB_a}{dx^j}=e^{1-d/t}P_j(t^{-1}),
  \label{eq_flat-derivative-form}
\end{equation}
where $P_j$ is a polynomial of degree at most $2j$ whose coefficients depend on $d$. Indeed, $d/dx=-d/dt$, and each differentiation of the right-hand side of Eq. \eqref{eq_flat-derivative-form} preserves the exponential factor while increasing the polynomial degree by at most two. Hence, for sufficiently small $t>0$, a constant $C_j>0$ exists such that
\begin{equation}
  \left|\frac{d^jB_a}{dx^j}\right|
  \leq C_j t^{-2j}e^{-d/t}.
  \label{eq_flat-estimate}
\end{equation}
After the substitution $y=d/t$, the right-hand side becomes a constant multiple of $y^{2j}e^{-y}$, which tends to zero as $y\to\infty$. Therefore,
\begin{equation}
  \lim_{x\to x_c^-}\frac{d^jB_a}{dx^j}=0,
  \qquad j\in\mathbb N_0.
  \label{eq_flat-all-derivatives}
\end{equation}
The right branch and all its derivatives vanish identically. The two branches therefore join in $C^\infty$, which proves Eq. \eqref{eq_complete-jet-agreement}. This is the standard flat-function mechanism used to construct smooth cutoffs \cite{Lee2012}.

The same argument establishes nonanalyticity, rather than merely failing to establish analyticity. Every derivative of $B_a$ at $x_c$ vanishes, so its Taylor series there is identically zero. Yet $B_a(x)>0$ for every $a\leq x<x_c$, whereas a real-analytic function must agree locally with its Taylor series. Thus $B_a$ is not analytic at $x_c$. Since $F_u(a)>0$, the difference $F_a-F_u=-F_u(a)B_a$ is nonzero immediately to the left of $x_c$ but has a zero Taylor series there. The local Taylor characterization of real-analytic functions \cite{KrantzParks2002} therefore implies that $F_a$ is nonanalytic at $x_c$.

Because $F_a=F_u$ for every $x\geq x_c$, the asymptotic normalization, ADM mass coefficient, charge coefficient, and entire outer part of the scalar potential remain unchanged. Between the new horizon and the photon sphere,
\begin{equation}
  U_a(x)=U_u(x)-\frac{F_u(a)}{x^2}B_a(x)<U_u(x),
  \qquad a<x<x_c,
  \label{eq_smooth-potential-ordering}
\end{equation}
whereas $U_a=U_u$ for $x\geq x_c$. Section~\ref{sec_analytic-nonuniqueness} showed that $U_u$ increases to its unique nondegenerate exterior maximum at $x_c$ and decreases thereafter. Eq. \eqref{eq_smooth-potential-ordering} therefore leaves $x_c$ as the unique global exterior maximum of $U_a$, and Eq. \eqref{eq_complete-jet-agreement} preserves its nondegeneracy.

Both lapses are $C^\infty$ near $x_c$ and share the complete jet there. The infinite-jet statement in Theorem~\ref{thm:photon-jet-locality} then gives Eq. \eqref{eq_all-orders-eikonal-agreement} for every fixed $n$. This step uses only the formal normal-form recursion proved in Sec. \ref{sec_jet-locality}. It does not replace the global quasinormal-mode boundary-value problem with a local one.

With the common asymptotic time normalization, the Hawking relation used in Sec. \ref{sec_analytic-nonuniqueness} is $T_H=F'(x_h)/(4\pi M)$. Substituting Eq. \eqref{eq_smooth-simple-horizon} gives Eq. \eqref{eq_smooth-temperature}. The horizon two-sphere has areal radius $Ma$ and area $A_H=4\pi M^2a^2$, so the Einstein area law yields Eq. \eqref{eq_smooth-entropy}. To verify that these quantities vary within the family, use $F_u(x_+)=0$ and $F_u'(x_+)=(x_+-x_-)/x_+^2$. The one-sided limit $a\to x_+^+$ then gives the two normalized limits
\begin{subequations}
\label{eq_smooth-RN-ratios}
\begin{align}
  \frac{T_H(a)}{T_H^{\rm RN}}
  &=\frac{F_u'(a)+F_u(a)/(x_c-a)}{F_u'(x_+)}\longrightarrow1,
  \label{eq_smooth-temperature-RN-ratio}\\
  \frac{S_{\rm EH}(a)}{S_{\rm RN}}&=\left(\frac{a}{x_+}\right)^2
  \longrightarrow1.
  \label{eq_smooth-entropy-RN-ratio}
\end{align}
\end{subequations}
By contrast, $F_u(a)/(x_c-a)\to+\infty$ as $a\to x_c^-$, while $F_u'(a)$ remains finite. This proves Eqs.~\eqref{eq_smooth-temperature-limits} and \eqref{eq_smooth-entropy-variation} and completes the proof.

We can describe the family without a profile plot by parametrizing the moving horizon as
\begin{equation}
  \eta=\frac{a-x_+}{x_c-x_+},
  \qquad 0\leq\eta<1,
  \label{eq_eta-smooth-family}
\end{equation}
where $\eta=0$ is the RN limit and $\eta\to1^-$ moves the horizon toward the unchanged photon sphere. Eq. \eqref{eq_smooth-potential-ordering} already contains all profile information needed for the proof. The barrier is lowered only on $a<x<x_c$ and remains unchanged on $x\geq x_c$. Likewise, the exact entropy variation is
\begin{equation}
  \frac{S_{\rm EH}(\eta)}{S_{\rm RN}}
  =\left[1+\eta\frac{x_c-x_+}{x_+}\right]^2,
  \label{eq_smooth-entropy-eta}
\end{equation}
while Eq. \eqref{eq_smooth-temperature} gives a finite RN limit and the divergence $T_H\to+\infty$ as $\eta\to1^-$. A graph would only reproduce the monotonic area law and the endpoint behavior already proved, so we omit a thermodynamic profile.

For reference, using $x_c^2-3x_c+2u=0$ reduces the common leading photon-sphere invariants to
\begin{subequations}
\label{eq_RN-photon-invariants-section4}
\begin{align}
  \Omega_c^2&=\frac{x_c-1}{2x_c^3},
  \label{eq_RN-photon-frequency-section4}\\
  \Lambda_c^2&=\frac{(x_c-1)(2x_c-3)}{2x_c^4},
  \label{eq_RN-Lyapunov-section4}\\
  \frac{\Omega_c}{\Lambda_c}&=\sqrt{\frac{x_c}{2x_c-3}}.
  \label{eq_RN-invariant-ratio-section4}
\end{align}
\end{subequations}
The leading quality factor is $\mathcal Q_{Ln}^{(0)}=L\Omega_c/[2(n+1/2)\Lambda_c]$. None of these quantities depends on $a$. Their normalized values are therefore identically one throughout the family, so plotting them would add no information to Eq. \eqref{eq_RN-photon-invariants-section4}.

The all-orders statement remains a formal fixed-overtone, large-$L$ result. The potential and the horizon boundary condition differ on the inner side of the common barrier, so the complete scalar boundary-value problems are not identical. Exact quasinormal frequencies are global scattering resonances \cite{BertiCardosoStarinets2009,DyatlovZworski2019}. The common algebraic series does not control their possible difference, which may lie beyond all orders in $L^{-1}$. We assert no exact finite-$\ell$ isospectrality. The construction is also geometric, and Sec. \ref{sec_effective-source} separately examines whether the effective Einstein source obeys an energy condition.

\section{Weak-energy-compatible effective-source domain}
\label{sec_effective-source}

Section~\ref{sec_smooth-nonuniqueness} constructed the family at the metric level. We now ask whether a nonempty part of that family can be interpreted as an Einstein geometry whose effective source satisfies the weak energy condition. This interpretation is deliberately limited. The tensor defined below is an effective Einstein source, from which we infer no unique microscopic matter Lagrangian.

We derive the source for a general lapse rather than importing a specialized formula. For the dimensionless coordinates $X^\mu=(\bar t,x,\theta,\phi)$, the nonzero independent Christoffel symbols of Eq. \eqref{eq_metric-section2} are most compactly specified by the following ordered component tuples.
\begin{align}
  \left(\Gamma^{\bar t}_{\bar t x},
  \Gamma^x_{\bar t\bar t},\Gamma^x_{xx}\right)
  &=\left(\frac{F'}{2F},\frac{FF'}{2},-\frac{F'}{2F}\right),
  \label{eq_Christoffel-radial-section5}\\
  \left(\Gamma^x_{\theta\theta},\Gamma^x_{\phi\phi},
  \Gamma^\theta_{x\theta},\Gamma^\phi_{x\phi}\right)
  &=\left(-xF,-xF\sin^2\theta,\frac1x,\frac1x\right),
  \label{eq_Christoffel-angular-section5}\\
  \left(\Gamma^\theta_{\phi\phi},\Gamma^\phi_{\theta\phi}\right)
  &=\left(-\sin\theta\cos\theta,\cot\theta\right).
  \label{eq_Christoffel-sphere-section5}
\end{align}
Symmetry supplies the lower-index permutations. Direct contraction of the Riemann tensor yields
\begin{align}
  \left(R_{\bar t\bar t},R_{xx}\right)
  &=\left(\frac{F}{2}\left[F''+\frac{2F'}x\right],
  -\frac{1}{2F}\left[F''+\frac{2F'}x\right]\right),
  \label{eq_Ricci-radial-section5}\\
  \left(R_{\theta\theta},R_{\phi\phi}\right)
  &=\left(1-F-xF',(1-F-xF')\sin^2\theta\right),
  \label{eq_Ricci-angular-section5}\\
  M^2R&=-F''-\frac{4F'}x+\frac{2(1-F)}{x^2}.
  \label{eq_Ricci-scalar-section5}
\end{align}
The constant overall factor $M^2$ in the metric produces the factor $M^{-2}$ in the scalar curvature. The resulting mixed Einstein tensor is
\begin{equation}
  M^2G^\mu{}_\nu=
  \operatorname{diag}\!\left(
    \frac{F+xF'-1}{x^2},
    \frac{F+xF'-1}{x^2},
    \frac{F''}{2}+\frac{F'}x,
    \frac{F''}{2}+\frac{F'}x
  \right).
  \label{eq_Einstein-tensor-section5}
\end{equation}
An independent symbolic contraction of the metric, connection, and curvature tensors reproduces these expressions.

We impose $G^\mu{}_\nu=8\pi G_NT^\mu{}_\nu$ and write the static orthonormal eigenvalues as
\begin{equation}
  T^\mu{}_\nu=\operatorname{diag}(-\rho,p_r,p_t,p_t).
  \label{eq_anisotropic-source-section5}
\end{equation}
Eq. \eqref{eq_Einstein-tensor-section5} gives
\begin{align}
  8\pi G_NM^2\rho&=\frac{1-F-xF'}{x^2},
  \label{eq_rho-general-section5}\\
  p_r&=-\rho,
  \label{eq_pr-general-section5}\\
  8\pi G_NM^2p_t&=\frac{F''}{2}+\frac{F'}x.
  \label{eq_pt-general-section5}
\end{align}
The radial component of stress-tensor conservation reduces to
\begin{equation}
  p_r'+\frac{F'}{2F}(\rho+p_r)+\frac2x(p_r-p_t)=0.
  \label{eq_conservation-section5}
\end{equation}
Substituting Eqs.~\eqref{eq_rho-general-section5}--\eqref{eq_pt-general-section5} makes the left-hand side vanish identically.

For a diagonal type-I source, the weak energy condition requires $\rho\geq0$ and $\rho+p_i\geq0$ \cite{HawkingEllis1973}. Eq. \eqref{eq_pr-general-section5} saturates the radial null condition, leaving only two independent margins. Define
\begin{align}
  \mathcal R[F]&:=8\pi G_NM^2\rho
  =\frac{1-F-xF'}{x^2},
  \label{eq_R-margin-section5}\\
  \mathcal N[F]&:=8\pi G_NM^2(\rho+p_t)
  =\frac{1-F}{x^2}+\frac{F''}{2}.
  \label{eq_N-margin-section5}
\end{align}
For the RN reference, the ordered triple consisting of the energy-density margin, tangential pressure, and tangential null margin reduces to
\begin{equation}
  \left(\mathcal R[F_u],8\pi G_NM^2p_{t,u},\mathcal N[F_u]\right)
  =\frac{u}{x^4}(1,1,2).
  \label{eq_RN-margins-section5}
\end{equation}
Their strict positivity for $u>0$ supplies the margin used below and explains why the uniform theorem excludes the uncharged endpoint.

On $a\leq x<x_c$, let $A:=F_u(a)>0$ be the deformation amplitude, $d:=x_c-a$ the photon-sphere--horizon separation, and $t:=x_c-x$ the distance to the photon sphere. The cutoff is then $B_a=e^{1-d/t}$.
The source depends on the derivatives
\begin{equation}
  \left(B_a',B_a''\right)
  =\left(-\frac{d}{t^2}B_a,\frac{d(d-2t)}{t^4}B_a\right).
  \label{eq_bump-derivatives-section5}
\end{equation}
Substituting $F_a=F_u-AB_a$ into Eqs.~\eqref{eq_R-margin-section5} and \eqref{eq_N-margin-section5} gives the closed deformation formulas
\begin{align}
  \mathcal R[F_a]
  &=\frac{u}{x^4}
  +\frac{AB_a}{x^2}\left(1-\frac{xd}{t^2}\right),
  \label{eq_rho-deformation-section5}\\
  \mathcal N[F_a]
  &=\frac{2u}{x^4}
  +AB_a\left[\frac1{x^2}-\frac{d(d-2t)}{2t^4}\right].
  \label{eq_N-deformation-section5}
\end{align}
For $x\geq x_c$, the lapse and all its derivatives agree exactly with RN, so Eq. \eqref{eq_RN-margins-section5} holds without an estimate. Although individual terms in Eqs.~\eqref{eq_rho-deformation-section5} and \eqref{eq_N-deformation-section5} contain inverse powers of $t$, multiplication by $B_a$ makes them vanish as $t\to0^+$, consistently with the complete matching proved in Sec. \ref{sec_smooth-nonuniqueness}.

\begin{theorem}[Effective-source admissibility]
  \label{thm:effective-source-admissibility}
  Let $K=[u_0,u_1]\Subset(0,1)$ and define
  \begin{align}
    (r_0,X)&=\left(1+\sqrt{1-u_1},
    \frac{3+\sqrt{9-8u_0}}2\right),
    \label{eq_compact-radii-section5}\\
    D&=\min_{u\in K}\left[
    \frac{3+\sqrt{9-8u}}2-1-\sqrt{1-u}\right]>0,
    \label{eq_D-definition-section5}\\
    (C_1,C_2,L_K)&=\left(\frac4e,
    \frac{256}{e^3}+\frac{54}{e^2},\frac{2}{r_0^2}\right),
    \label{eq_uniform-constants-section5}\\
    (C_{\mathcal R},C_{\mathcal N})
    &=L_K\left(\frac1{r_0^2}+\frac{2C_1}{Dr_0},
    \frac1{r_0^2}+\frac{2C_2}{D^2}\right),
    \label{eq_margin-constants-section5}\\
    \delta_K&=\min\left\{
      \frac D2,\frac{u_0}{2X^4C_{\mathcal R}},
      \frac{u_0}{X^4C_{\mathcal N}}
    \right\}>0.
    \label{eq_deltaK-section5}
  \end{align}
  Then every member of the smooth family that satisfies
  \begin{equation}
    \bigl(u,a-x_+(u)\bigr)\in K\times[0,\delta_K)
    \label{eq_WEC-wedge-section5}
  \end{equation}
  obeys the weak energy condition throughout its exterior $x\geq a$. The endpoint $a=x_+(u)$ is the undeformed RN geometry. Every nontrivial deformation has $a>x_+(u)$.
\end{theorem}

\textit{Proof.} Write $s=a-x_+(u)$. Since $s<D/2$, we have $d=x_c-a\geq D/2$. With $z:=d/t\geq1$, the definitions above give
\begin{equation}
  \left(|B_a'|,|B_a''|\right)
  =\left(\frac{z^2e^{1-z}}d,
  \frac{z^3|z-2|e^{1-z}}{d^2}\right).
  \label{eq_z-bump-bounds-section5}
\end{equation}
The first numerator has maximum $4/e$ on $z\geq1$. For the second, use $z^3|z-2|\leq z^4+2z^3$ and $\sup_{z\geq1}z^me^{1-z}=m^me^{1-m}$ for $m=3,4$. Therefore,
\begin{equation}
  \left(B_a,|B_a'|,|B_a''|\right)
  \in[0,1]\times\left[0,\frac{C_1}{d}\right]
  \times\left[0,\frac{C_2}{d^2}\right].
  \label{eq_uniform-bump-bounds-section5}
\end{equation}
These analytic bounds hold on the entire deformed interval and require no radial sampling.

For $x\geq x_+(u)\geq r_0$, the RN derivative satisfies $0<F_u'(x)=2(x-u)/x^3\leq2/r_0^2=L_K$. The mean-value theorem therefore gives
\begin{equation}
  0\leq A=F_u(a)-F_u(x_+)\leq L_Ks.
  \label{eq_A-linear-bound-section5}
\end{equation}
Moreover, $a\leq x<x_c\leq X$. Equations~\eqref{eq_uniform-bump-bounds-section5} and \eqref{eq_A-linear-bound-section5} imply
\begin{align}
  |\mathcal R[F_a]-\mathcal R[F_u]|
  &\leq C_{\mathcal R}s,
  \label{eq_R-uniform-bound-section5}\\
  |\mathcal N[F_a]-\mathcal N[F_u]|
  &\leq C_{\mathcal N}s.
  \label{eq_N-uniform-bound-section5}
\end{align}
On the same interval, the RN margins obey the componentwise bounds
\begin{equation}
  \left(\mathcal R[F_u],\mathcal N[F_u]\right)
  \in\left[\frac{u_0}{X^4},\infty\right)
  \times\left[\frac{2u_0}{X^4},\infty\right).
  \label{eq_RN-uniform-lower-bounds-section5}
\end{equation}
The definition of $\delta_K$ then gives
\begin{equation}
  \left(\mathcal R[F_a],\mathcal N[F_a]\right)
  \in\left[\frac{u_0}{2X^4},\infty\right)
  \times\left[\frac{u_0}{X^4},\infty\right)
  \label{eq_deformed-positive-margins-section5}
\end{equation}
for $a\leq x<x_c$. The region $x\geq x_c$ is exactly RN and has positive margins by Eq. \eqref{eq_RN-margins-section5}. Together with the saturated radial null condition, this proves the theorem.

The minimum in Eq. \eqref{eq_D-definition-section5} is elementary. One checks the two endpoints of $K$ and, if it lies in $K$, the sole interior stationary point $u=7/8$. Eq. \eqref{eq_deltaK-section5} is therefore a closed, computable sufficient bound rather than an existence constant obtained by numerical search. The bound is intentionally conservative because the theorem requires a certified nonempty wedge, not the largest possible one. A plot would only display the imposed linear estimates, so we state the admissible domain analytically.

The dominant and strong energy conditions require separate analysis. The strong condition adds $\rho+p_r+2p_t=2p_t\geq0$. From
\begin{equation}
  8\pi G_NM^2[p_t(F_a)-p_t(F_u)]
  =-A\left(\frac{B_a''}{2}+\frac{B_a'}x\right),
  \label{eq_pt-deformation-section5}
\end{equation}
the preceding estimates show that the strong condition holds throughout a possibly smaller, but still explicit, wedge. Define $C_t:=L_K[2C_2/D^2+2C_1/(Dr_0)]$. The additional sufficient bound is
\begin{equation}
  a-x_+<\frac{u_0}{2X^4C_t}
  \label{eq_SEC-bound-section5}
\end{equation}
to Eq. \eqref{eq_WEC-wedge-section5}. This conclusion follows from a separate estimate, not from the weak condition.

By contrast, the dominant condition also requires $\rho-p_t\geq0$. At the new horizon, where $B_a(a)=1$, $B_a'(a)=-1/d$, and $B_a''(a)=-1/d^2$, we obtain
\begin{equation}
  8\pi G_NM^2(\rho-p_t)\big|_{x=a}
  =A\left(\frac1{a^2}-\frac{2}{ad}-\frac1{2d^2}\right)<0.
  \label{eq_DEC-violation-section5}
\end{equation}
Here $A>0$ for every $a>x_+$, and $0<d<a$ follows from $x_c<2x_+<2a$ for $0\leq u<1$. The bracket is therefore strictly negative. Every nontrivial smooth deformation consequently violates the dominant energy condition at its horizon, even within the weak-energy wedge. This explicit limitation does not affect Theorem~\ref{thm:effective-source-admissibility}, which asserts only weak-energy admissibility. The theorem also makes no claim at $u=0$, where the strict RN margins in Eq. \eqref{eq_RN-uniform-lower-bounds-section5} vanish.

\section{Scalar-operator equivalence and action-dependent entropy}
\label{sec_operator-action}

Section~\ref{sec_effective-source} fixed Einstein gravity and asked whether the smoothly deformed geometries admit an effective source that satisfies the weak energy condition. We now consider the complementary obstruction. We hold the geometry and the complete neutral-scalar boundary-value problem fixed while changing the gravitational action. This comparison separates quantities determined by the normalized metric from entropy, which also depends on the coefficient of the curvature term in the action.

We rederive the RN--MOG metric and entropy relations, including the action-level Noether-charge step, because they supply the third information level used in Sec. \ref{sec_identifiability-hierarchy}. This section therefore provides a previously established ingredient of the present hierarchy, not a new RN--MOG result.

We use the Einstein--Maxwell action and the massless-vector, constant-coupling sector of scalar--tensor--vector gravity (MOG). For a constant MOG parameter $\alpha\geq0$, define the effective curvature coupling $G_\alpha:=G_N(1+\alpha)$. With the signature of Eq. \eqref{eq_metric-section2}, the relevant terms are
\begin{align}
  I_{\rm RN}
  &=\int d^4x\,\sqrt{-g}\left[
    \frac{R}{16\pi G_N}
  -\frac{\mathcal F_{ab}\mathcal F^{ab}}{16\pi}\right],
  \label{eq_RN-action-section6}\\
  I_{\alpha}
  &=\int d^4x\,\sqrt{-g}\left[
    \frac{R}{16\pi G_\alpha}
  -\frac{B_{ab}B^{ab}}{16\pi}\right].
  \label{eq_MOG-action-section6}
\end{align}
Here $\mathcal F_{ab}$ is the electromagnetic field strength, while $B_{ab}$ is the MOG vector-field strength. In Eq. \eqref{eq_MOG-action-section6}, $\alpha\geq0$ is constant, the MOG scalar couplings are frozen, the vector mass and cosmological constant are neglected, and the vector kinetic term has the Maxwell normalization. These assumptions define the singular static MOG black-hole solution used here \cite{Moffat2015MOGBH,MureikaMoffatFaizal2016}. Terms in the frozen scalar sector that contain no Riemann tensor do not affect the entropy calculation below.

Let $\mathcal M_\alpha$ denote the physical MOG mass and define the length $m_\alpha=G_N\mathcal M_\alpha$. The MOG vector charge and lapse are
\begin{align}
  Q_g&=\sqrt{\alpha G_N}\,\mathcal M_\alpha,
  \label{eq_MOG-charge-section6}\\
  F_\alpha(r)
  &=1-\frac{2G_\alpha\mathcal M_\alpha}{r}
  +\frac{G_\alpha Q_g^2}{r^2}\nonumber\\
  &=1-\frac{2(1+\alpha)m_\alpha}{r}
  +\frac{\alpha(1+\alpha)m_\alpha^2}{r^2}.
  \label{eq_MOG-lapse-section6}
\end{align}
The physical mass differs from the coefficient measured in the asymptotic metric. Introduce the common geometric mass $M$ and normalized charge parameter $u$ by
\begin{subequations}
\label{eq_MOG-RN-map-section6}
\begin{align}
  M&=G_\alpha\mathcal M_\alpha=(1+\alpha)m_\alpha,
  \label{eq_MOG-geometric-mass-map}\\
  u&=\frac{\alpha}{1+\alpha}.
  \label{eq_MOG-charge-parameter-map}
\end{align}
\end{subequations}
The inverse relation is $\alpha=u/(1-u)$.
For every finite $\alpha\geq0$, we have $0\leq u<1$, and Eq. \eqref{eq_MOG-lapse-section6} becomes
\begin{equation}
  F_\alpha(r)
  =1-\frac{2M}{r}+\frac{uM^2}{r^2}
  =F_u\!\left(\frac rM\right).
  \label{eq_MOG-RN-metric-identity-section6}
\end{equation}
An Einstein--Maxwell RN solution with physical mass $\mathcal M_{\rm RN}=M/G_N$ and electric charge normalized by $G_NQ_e^2=uM^2$ has exactly the same lapse. Thus the two theories assign different physical masses and fields to the same normalized geometry. The comparison does not identify the theories themselves.

\begin{corollary}[Scalar-operator/action non-identifiability]
  \label{cor:operator-action-nonidentifiability}
  Fix $\alpha>0$ and impose the map in Eq. \eqref{eq_MOG-RN-map-section6}.
  The comparison holds the common geometric length scale $M$ fixed. The physical masses are $\mathcal M_{\rm RN}=M/G_N$ and $\mathcal M_\alpha=M/G_\alpha=M/[G_N(1+\alpha)]$, which differ for $\alpha>0$.
  The RN solution of Eq. \eqref{eq_RN-action-section6} and the constant-coupling MOG solution of Eq. \eqref{eq_MOG-action-section6} have the same areal radius, asymptotic time normalization, lapse, and horizons. They also define the same boundary-value problem for a neutral, minimally coupled, massless scalar. Their scalar quasinormal frequencies therefore satisfy
  \begin{equation}
    \omega_{\ell n}^{\rm MOG}
    =\omega_{\ell n}^{\rm RN}
    \qquad
    \text{for every allowed }(\ell,n),
    \label{eq_exact-scalar-isospectrality-section6}
  \end{equation}
  without an eikonal or fixed-overtone approximation. Their Hawking temperatures are also equal, but their Wald entropies obey
  \begin{equation}
    \frac{S_{\rm MOG}}{S_{\rm RN}}
    =\frac{G_N}{G_\alpha}
    =\frac{1}{1+\alpha}
    =1-u<1.
    \label{eq_entropy-ratio-section6}
  \end{equation}
  Consequently, even the complete neutral-scalar operator and its exact quasinormal spectrum do not determine Wald entropy unless the gravitational action is supplied independently.
\end{corollary}

\textit{Proof.} Eq. \eqref{eq_MOG-RN-metric-identity-section6} gives the common line element and lapse
\begin{subequations}
\label{eq_common-metric-section6}
\begin{align}
  ds^2=-F(r)dt^2+\frac{dr^2}{F(r)}+r^2d\Omega_2^2,
  \label{eq_common-line-element-section6}\\
  F(r)=\frac{(r-r_+)(r-r_-)}{r^2}.
  \label{eq_common-factorized-lapse-section6}
\end{align}
\end{subequations}
where
\begin{equation}
  r_\pm=M\left(1\pm\sqrt{1-u}\right)
  =m_\alpha\left(1+\alpha\pm\sqrt{1+\alpha}\right).
  \label{eq_common-horizons-section6}
\end{equation}
Both theories use the same areal coordinate $r$ and the same Killing time, normalized by $g_{tt}\to-1$ at infinity. If $r_*^{\rm RN}$ and $r_*^{\rm MOG}$ denote their tortoise coordinates, then
\begin{equation}
  \frac{dr_*^{\rm RN}}{dr}
  =\frac1{F(r)}
  =\frac{dr_*^{\rm MOG}}{dr}
  \label{eq_tortoise-identity-section6}
\end{equation}
and hence $r_*^{\rm MOG}=r_*^{\rm RN}+C$ for some constant $C$.

For an external neutral test scalar, both theories impose $\Box\Phi=0$. Repeating the separation used in Eqs.~\eqref{eq_scalar-separation} and \eqref{eq_scalar-radial}, now in physical coordinates, gives the common radial equation and potential
\begin{subequations}
\label{eq_common-scalar-problem-section6}
\begin{align}
  \frac{d^2\psi}{dr_*^2}
  +\left[\omega^2-V_\ell(r)\right]\psi=0,
  \label{eq_common-radial-equation-section6}\\
  V_\ell(r)=F(r)\left[
  \frac{\ell(\ell+1)}{r^2}+\frac1r\frac{dF}{dr}\right].
  \label{eq_common-scalar-potential-section6}
\end{align}
\end{subequations}
Equality of $F$ makes both the differential expression and the potential identical for every $\ell$. The additive constant in Eq. \eqref{eq_tortoise-identity-section6} changes only the normalization of the asymptotic wave factors. With the $e^{-i\omega t}$ convention, the quasinormal conditions are
\begin{subequations}
\label{eq_common-QNM-boundaries-section6}
\begin{align}
  \psi&\sim e^{-i\omega r_*}\quad(r_*\to-\infty),
  \label{eq_ingoing-QNM-boundary-section6}\\
  \psi&\sim e^{+i\omega r_*}\quad(r_*\to+\infty).
  \label{eq_outgoing-QNM-boundary-section6}
\end{align}
\end{subequations}
These conditions describe ingoing waves at the outer horizon and outgoing waves at infinity, respectively \cite{BertiCardosoStarinets2009}. Hence the two radial differential operators have the same domain and resonance condition. Eq. \eqref{eq_exact-scalar-isospectrality-section6} follows from an exact operator identity, not from a numerical comparison or a large-$L$ expansion.

The common normalized metric also fixes the surface gravity.
\begin{equation}
  \kappa
  =\frac12\left.\frac{dF}{dr}\right|_{r_+}
  =\frac{r_+-r_-}{2r_+^2}.
  \label{eq_common-surface-gravity-section6}
\end{equation}
Thus both theories yield
\begin{align}
  T_H
  &=\frac{\kappa}{2\pi}
  =\frac{\sqrt{1-u}}
  {2\pi M\left(1+\sqrt{1-u}\right)^2}
  \nonumber\\
  &=\frac{1}
  {2\pi m_\alpha\sqrt{1+\alpha}
  \left(1+\sqrt{1+\alpha}\right)^2}.
  \label{eq_common-temperature-section6}
\end{align}
The temperatures are equal for geometric reasons. The lapse, horizon, and time normalization are identical.

Entropy also requires the action. Let $\varepsilon_{ab}$ be the horizon binormal normalized by $\varepsilon_{ab}\varepsilon^{ab}=-2$. For a diffeomorphism-invariant scalar Lagrangian $\mathcal L$ with no derivatives of the Riemann tensor, the stationary Noether-charge entropy is
\begin{equation}
  S_{\rm W}
  =-2\pi\int_{\mathcal H}
  \frac{\partial\mathcal L}{\partial R_{abcd}}\,
  \varepsilon_{ab}\varepsilon_{cd}\sqrt{h}\,d^2x
  \label{eq_Wald-formula-section6}
\end{equation}
\cite{Wald1993,IyerWald1994}. For either action in Eqs.~\eqref{eq_RN-action-section6} and \eqref{eq_MOG-action-section6}, let $G_T$ denote the corresponding curvature coupling, $G_T=G_N$ or $G_T=G_\alpha$. Direct differentiation gives
\begin{equation}
  \frac{\partial\mathcal L_T}{\partial R_{abcd}}
  =\frac{1}{32\pi G_T}
  \left(g^{ac}g^{bd}-g^{ad}g^{bc}\right).
  \label{eq_curvature-derivative-section6}
\end{equation}
The electromagnetic and MOG vector terms contain no Riemann tensor and therefore do not contribute to Eq. \eqref{eq_curvature-derivative-section6}. Using
\begin{equation}
  \left(g^{ac}g^{bd}-g^{ad}g^{bc}\right)
  \varepsilon_{ab}\varepsilon_{cd}=-4
  \label{eq_binormal-contraction-section6}
\end{equation}
in Eq. \eqref{eq_Wald-formula-section6} gives, with $A_H:=4\pi r_+^2$ denoting the common horizon area,
\begin{equation}
  S_{\rm W}=\frac{A_H}{4G_T}.
  \label{eq_Wald-area-section6}
\end{equation}
The common geometry fixes $A_H$, while the two actions give
\begin{align}
  S_{\rm RN}
  &=\frac{\pi r_+^2}{G_N}
  =\frac{\pi M^2\left(1+\sqrt{1-u}\right)^2}{G_N},
  \label{eq_RN-Wald-entropy-section6}\\
  S_{\rm MOG}
  &=\frac{\pi r_+^2}{G_\alpha}
  =\pi G_N\mathcal M_\alpha^2
  \left(1+\sqrt{1+\alpha}\right)^2.
  \label{eq_MOG-Wald-entropy-section6}
\end{align}
Dividing these expressions proves Eq. \eqref{eq_entropy-ratio-section6}. This derivation does not use first-law integration. The unequal entropies follow directly from the unequal curvature coefficients. Ref. \cite{MureikaMoffatFaizal2016} also reports an entropy obtained by integrating $d\mathcal M_\alpha/T_H$ and distinguishes it from the area expression. Along that integration, however, $Q_g=\sqrt{\alpha G_N}\,\mathcal M_\alpha$ varies with $\mathcal M_\alpha$, so the corresponding vector-charge work term cannot consistently be omitted from the first law. The entropy relevant to the present comparison is the Wald entropy derived directly from the stated constant-coupling action.

The exact isospectrality in Eq. \eqref{eq_exact-scalar-isospectrality-section6} is stronger than the formal eikonal equivalences in Secs.~\ref{sec_analytic-nonuniqueness} and \ref{sec_smooth-nonuniqueness}, but it applies to a narrower field content. A charged probe distinguishes the electromagnetic and MOG vector couplings. A nonminimally coupled scalar can depend on theory-specific couplings even on the same metric. Metric, electromagnetic, MOG-vector, and MOG-scalar perturbations obey different coupled systems and are not covered by Eq. \eqref{eq_exact-scalar-isospectrality-section6}. The corollary asserts neither equivalence of the full theories nor gravitational-wave isospectrality. It states only that an action-blind neutral test field cannot identify an action-dependent entropy.

\section{Identifiability hierarchy and conclusions}
\label{sec_identifiability-hierarchy}

The preceding results separate three inverse problems that are otherwise easy to conflate. Recall that $\mathcal D_N[F]$ denotes the formal scalar eikonal data through order $L^{-N}$ at fixed overtone, as defined before Eq. \eqref{eq_eikonal-functional-introduction}, and let $\mathcal D_\infty[F]$ denote the corresponding complete formal series. We also write $\mathcal O_\phi[F]$ for the neutral, minimally coupled, massless scalar differential operator together with its quasinormal boundary conditions. These objects contain successively more information, but none removes every thermodynamic ambiguity.

The first two levels below follow from the new finite-jet and smooth counterexamples. The third incorporates the RN--MOG operator/action example rederived in Sec. \ref{sec_operator-action}. The new result of this section is the hierarchy obtained by placing these logically distinct inputs on a single information scale.

\begin{corollary}[Three-level identifiability hierarchy]
  \label{cor:identifiability-hierarchy}
  Under the hypotheses of Theorems~\ref{thm:photon-jet-locality}--\ref{thm:effective-source-admissibility} and Corollary~\ref{cor:operator-action-nonidentifiability}, the following statements hold.

  First, $\mathcal D_N$ factors through the $(2N+2)$-jet of the lapse at the photon sphere. Its fibers in the real-analytic, asymptotically RN class contain metrics with different Hawking temperatures. A finite formal eikonal truncation therefore does not determine $T_H$.

  Second, $\mathcal D_\infty$ factors through the complete $C^\infty$ jet at the photon sphere. Its fibers in the smooth metric class contain geometries with different horizon positions, Hawking temperatures, and Einstein area entropies. This nonuniqueness persists in a nonempty subfamily whose effective Einstein source satisfies the weak energy condition.

  Third, $\mathcal O_\phi[F]$ fixes the exact neutral-scalar resonance set. If the normalized metric is also supplied, it fixes the geometric horizon position, surface gravity, and Hawking temperature as well. Nevertheless, the same metric and scalar operator can arise from gravitational actions with different curvature couplings and hence different Wald entropies. The action must therefore be supplied independently before $S_{\rm W}$ is identifiable.
\end{corollary}

\textit{Proof.} Theorem~\ref{thm:photon-jet-locality} proves that every entry of $\mathcal D_N$ depends only on $j^{2N+2}_{x_c}F$. It does not assert that the coefficient tuple recovers the entire jet. For every $N$, Theorem~\ref{thm:analytic-finite-nonuniqueness} then constructs a one-parameter real-analytic family with the required common jet, and hence a common value of $\mathcal D_N$, while $T_H$ varies strictly with the deformation parameter. This proves the first level.

The complete-jet part of Theorem~\ref{thm:photon-jet-locality} makes $\mathcal D_\infty$ a local formal invariant of $j^\infty_{x_c}F$. Theorem~\ref{thm:smooth-all-orders} supplies a family with a common complete jet but variable $x_h$, $T_H$, and $S_{\rm EH}$. Theorem~\ref{thm:effective-source-admissibility} identifies a nonempty wedge of the same family that obeys the weak energy condition as an effective Einstein source. These results establish the second level without promoting formal eikonal equality to exact finite-$\ell$ isospectrality.

Finally, Corollary~\ref{cor:operator-action-nonidentifiability} gives two theories with the same normalized RN-form metric, the same operator $\mathcal O_\phi$, the same scalar boundary conditions, and therefore the same exact scalar spectrum. Their common metric gives the same $T_H$, whereas Eq. \eqref{eq_entropy-ratio-section6} gives $S_{\rm MOG}/S_{\rm RN}=1-u$. Thus exact scalar-operator data do not select the curvature coefficient in the action. This proves the third level and completes the hierarchy.

Table~\ref{tab:identifiability-hierarchy} summarizes the logical content of the corollary. The qualification at most is essential in the first two rows. The locality theorem bounds the information on which the coefficients can depend, but it does not reconstruct every component of a metric jet.

\begin{table}[t]
  \centering
  \small
  \setlength{\tabcolsep}{3pt}
  \renewcommand{\arraystretch}{1.22}
  \caption{Information fixed and left unfixed at the three data levels.}
  \label{tab:identifiability-hierarchy}
  \begin{tabular}{p{0.15\linewidth}p{0.26\linewidth}p{0.30\linewidth}p{0.17\linewidth}}
    \hline
    {\raggedright Data level\par}
    & {\raggedright Determined information\par}
    & {\raggedright Undetermined information\par}
    & {\raggedright Supporting result\par} \\
    \hline
    {\raggedright Finite formal eikonal truncation\par}
    & {\raggedright At most the local combinations contained in $j^{2N+2}_{x_c}F$\par}
    & {\raggedright The global exterior and $T_H$\par}
    & {\raggedright Theorems~\ref{thm:photon-jet-locality} and \ref{thm:analytic-finite-nonuniqueness}\par} \\
    {\raggedright Complete formal eikonal series\par}
    & {\raggedright At most the complete local normal-form germ at $x_c$\par}
    & {\raggedright $x_h$, $T_H$, and $S_{\rm EH}$ in the $C^\infty$ class\par}
    & {\raggedright Theorems~\ref{thm:smooth-all-orders} and \ref{thm:effective-source-admissibility}\par} \\
    {\raggedright Exact neutral-scalar operator and normalized metric\par}
    & {\raggedright Exact scalar resonances and the geometric $T_H$\par}
    & {\raggedright The gravitational action and $S_{\rm W}$\par}
    & {\raggedright Corollary~\ref{cor:operator-action-nonidentifiability}\par} \\
    \hline
  \end{tabular}
\end{table}

The hierarchy has a deliberately narrow scope. The probe is a neutral, minimally coupled, massless test scalar, and the first two levels concern the formal fixed-overtone limit $L\to\infty$. We do not claim exact finite-$\ell$ isospectrality for either the analytic family $F_{\epsilon,p}$ or the smooth family $F_a$. The all-orders family is $C^\infty$ but nonanalytic at $x_c$, and all geometries are static, spherically symmetric, asymptotically flat, and nonextremal. The source in Sec. \ref{sec_effective-source} is an effective Einstein source rather than a fundamental matter Lagrangian. Its uniform weak-energy result uses a compact charged interval $K\Subset(0,1)$, and every nontrivial smooth deformation still violates the dominant energy condition at its horizon. Rotation, extremality, excitation factors, detector inference, charged or nonminimal probes, and theory-specific gravitational, vector, or scalar perturbation systems lie outside the result.

Equality of all algebraic eikonal coefficients leaves a precise beyond-all-orders question open. If remainder estimates promote the two formal series to asymptotic expansions at every order, a possible exact-frequency splitting $\Delta\omega_{Ln}$ may satisfy
\begin{equation}
  \Delta\omega_{Ln}=o(L^{-q})
  \qquad\text{for every fixed }q>0
  \label{eq_beyond-all-orders-section7}
\end{equation}
without vanishing at finite $L$. An exponential term such as $e^{-\gamma L}$ has this property, but $C^\infty$ nonanalytic coefficients can produce a different superalgebraic scale. High-order WKB series and their Borel summation provide a related perspective on the asymptotic structure of black-hole quasinormal frequencies \cite{Hatsuda2020Borel}. The local barrier-top normal form cannot determine the global splitting scale because exact quasinormal frequencies follow from the connection problem between the horizon and infinity \cite{BertiCardosoStarinets2009,DyatlovZworski2019}. Thus beyond-all-orders sensitivity is neither hidden exact isospectrality nor a premise of the corollary. It is the part of the exact resonance problem absent from the local formal germ.

This boundary leads to two questions. First, can one establish uniform remainder estimates and quantify the global finite-$\ell$ splitting between jet-equivalent metrics, including whether its scale is exponential or merely superalgebraic in the relevant regularity class? Second, which additional perturbation channels distinguish actions that share a normalized metric and a neutral-scalar operator? These questions probe the two information losses established here without turning the inverse problem into detector parameter estimation.

These results replace a model-conditional reconstruction with a hierarchy of non-identifiability statements. A chosen metric family may provide an injective map from a few eikonal coefficients to its parameters, but that injectivity comes from the imposed global ansatz. Likewise, identifying the background metric does not allow an action-blind probe to determine the entropy formula. Thermodynamic reconstruction from eikonal ringdown is therefore conditional on assumptions about the global metric and the gravitational action. The local photon-sphere resonance germ alone is insufficient.

\section*{Data and code availability.}
No empirical data were generated or analyzed.

\begin{acknowledgments}
  N.J.L. Lobos and E.T. Rodulfo acknowledge De La Salle University and the DLSU Theoretical Physics Group for institutional support. N.J.L. Lobos also acknowledges support from the Department of Science and Technology--Accelerated Science and Technology Human Resource Development Program (DOST--ASTHRDP).
\end{acknowledgments}

\bibliographystyle{apsrev4-2}
\bibliography{ref}

@article{SchutzWill1985,
  author  = {Schutz, Bernard F. and Will, Clifford M.},
  title   = {Black hole normal modes: A semianalytic approach},
  journal = {The Astrophysical Journal Letters},
  volume  = {291},
  pages   = {L33--L36},
  year    = {1985},
  doi     = {10.1086/184453}
}

@article{IyerWill1987,
  author  = {Iyer, Sai and Will, Clifford M.},
  title   = {Black-hole normal modes: A {WKB} approach. I. Foundations and application of a higher-order {WKB} analysis of potential-barrier scattering},
  journal = {Physical Review D},
  volume  = {35},
  number  = {12},
  pages   = {3621--3631},
  year    = {1987},
  doi     = {10.1103/PhysRevD.35.3621}
}

@article{CardosoEtAl2009,
  author        = {Cardoso, Vitor and Miranda, Alex S. and Berti, Emanuele and Witek, Helvi and Zanchin, Vilson T.},
  title         = {Geodesic stability, Lyapunov exponents, and quasinormal modes},
  journal       = {Physical Review D},
  volume        = {79},
  number        = {6},
  pages         = {064016},
  year          = {2009},
  doi           = {10.1103/PhysRevD.79.064016},
  eprint        = {0812.1806},
  archivePrefix = {arXiv},
  primaryClass  = {hep-th}
}

@article{DolanOttewill2009,
  author        = {Dolan, Sam R. and Ottewill, Adrian C.},
  title         = {On an Expansion Method for Black Hole Quasinormal Modes and Regge Poles},
  journal       = {Classical and Quantum Gravity},
  volume        = {26},
  number        = {22},
  pages         = {225003},
  year          = {2009},
  doi           = {10.1088/0264-9381/26/22/225003},
  eprint        = {0908.0329},
  archivePrefix = {arXiv},
  primaryClass  = {gr-qc}
}

@article{Iantchenko2008,
  author  = {Iantchenko, Alexei},
  title   = {An Inverse Problem for Trapping Point Resonances},
  journal = {Letters in Mathematical Physics},
  volume  = {86},
  pages   = {151--157},
  year    = {2008},
  doi     = {10.1007/s11005-008-0276-1}
}

@article{Iantchenko2017,
  author        = {Iantchenko, Alexei},
  title         = {Quasi-normal modes for de Sitter--Reissner--Nordstr{\"o}m black holes},
  journal       = {Mathematical Research Letters},
  volume        = {24},
  number        = {1},
  pages         = {83--117},
  year          = {2017},
  doi           = {10.4310/MRL.2017.v24.n1.a5},
  eprint        = {1407.3654},
  archivePrefix = {arXiv},
  primaryClass  = {gr-qc}
}

@article{BertiCardosoStarinets2009,
  author        = {Berti, Emanuele and Cardoso, Vitor and Starinets, Andrei O.},
  title         = {Quasinormal modes of black holes and black branes},
  journal       = {Classical and Quantum Gravity},
  volume        = {26},
  number        = {16},
  pages         = {163001},
  year          = {2009},
  doi           = {10.1088/0264-9381/26/16/163001},
  eprint        = {0905.2975},
  archivePrefix = {arXiv},
  primaryClass  = {gr-qc}
}

@article{Bekenstein1973,
  author  = {Bekenstein, Jacob D.},
  title   = {Black Holes and Entropy},
  journal = {Physical Review D},
  volume  = {7},
  number  = {8},
  pages   = {2333--2346},
  year    = {1973},
  doi     = {10.1103/PhysRevD.7.2333}
}

@article{Hawking1975,
  author  = {Hawking, Stephen W.},
  title   = {Particle Creation by Black Holes},
  journal = {Communications in Mathematical Physics},
  volume  = {43},
  pages   = {199--220},
  year    = {1975},
  doi     = {10.1007/BF02345020}
}

@book{KrantzParks2002,
  author    = {Krantz, Steven G. and Parks, Harold R.},
  title     = {A Primer of Real Analytic Functions},
  edition   = {2},
  series    = {Birkh{\"a}user Advanced Texts Basler Lehrb{\"u}cher},
  publisher = {Birkh{\"a}user},
  address   = {Boston, MA},
  year      = {2002},
  doi       = {10.1007/978-0-8176-8134-0},
  isbn      = {978-0-8176-4264-8}
}

@book{HawkingEllis1973,
  author    = {Hawking, Stephen W. and Ellis, George F. R.},
  title     = {The Large Scale Structure of Space-Time},
  publisher = {Cambridge University Press},
  address   = {Cambridge},
  year      = {1973},
  doi       = {10.1017/CBO9780511524646}
}

@article{Wald1993,
  author        = {Wald, Robert M.},
  title         = {Black Hole Entropy Is the {Noether} Charge},
  journal       = {Physical Review D},
  volume        = {48},
  number        = {8},
  pages         = {R3427--R3431},
  year          = {1993},
  doi           = {10.1103/PhysRevD.48.R3427},
  eprint        = {gr-qc/9307038},
  archivePrefix = {arXiv},
  primaryClass  = {gr-qc}
}

@article{IyerWald1994,
  author        = {Iyer, Vivek and Wald, Robert M.},
  title         = {Some Properties of {Noether} Charge and a Proposal for Dynamical Black Hole Entropy},
  journal       = {Physical Review D},
  volume        = {50},
  number        = {2},
  pages         = {846--864},
  year          = {1994},
  doi           = {10.1103/PhysRevD.50.846},
  eprint        = {gr-qc/9403028},
  archivePrefix = {arXiv},
  primaryClass  = {gr-qc}
}

@article{VolkelKokkotas2019,
  author  = {V{\"o}lkel, Sebastian H. and Kokkotas, Kostas D.},
  title   = {Scalar Fields and Parametrized Spherically Symmetric Black Holes: Can One Hear the Shape of Space-Time?},
  journal = {Physical Review D},
  volume  = {100},
  number  = {4},
  pages   = {044026},
  year    = {2019},
  doi     = {10.1103/PhysRevD.100.044026}
}

@article{VolkelBarausse2020,
  author  = {V{\"o}lkel, Sebastian H. and Barausse, Enrico},
  title   = {Bayesian Metric Reconstruction with Gravitational Wave Observations},
  journal = {Physical Review D},
  volume  = {102},
  number  = {8},
  pages   = {084025},
  year    = {2020},
  doi     = {10.1103/PhysRevD.102.084025}
}

@article{VolkelFranchiniBarausse2022,
  author  = {V{\"o}lkel, Sebastian H. and Franchini, Nicola and Barausse, Enrico},
  title   = {Theory-Agnostic Reconstruction of Potential and Couplings from Quasinormal Modes},
  journal = {Physical Review D},
  volume  = {105},
  number  = {8},
  pages   = {084046},
  year    = {2022},
  doi     = {10.1103/PhysRevD.105.084046}
}

@article{VolkelFranchiniBarausseBerti2022,
  author  = {V{\"o}lkel, Sebastian H. and Franchini, Nicola and Barausse, Enrico and Berti, Emanuele},
  title   = {Constraining Modifications of Black Hole Perturbation Potentials near the Light Ring with Quasinormal Modes},
  journal = {Physical Review D},
  volume  = {106},
  number  = {12},
  pages   = {124036},
  year    = {2022},
  doi     = {10.1103/PhysRevD.106.124036}
}

@article{UhlmannWang2023,
  author  = {Uhlmann, Gunther and Wang, Yiran},
  title   = {Recovery of Black Hole Mass from a Single Quasinormal Mode},
  journal = {Communications in Mathematical Physics},
  volume  = {401},
  pages   = {925--936},
  year    = {2023},
  doi     = {10.1007/s00220-023-04666-0}
}

@article{Moffat2015MOGBH,
  author        = {Moffat, J. W.},
  title         = {Black Holes in Modified Gravity ({MOG})},
  journal       = {The European Physical Journal C},
  volume        = {75},
  number        = {4},
  pages         = {175},
  year          = {2015},
  doi           = {10.1140/epjc/s10052-015-3405-x},
  eprint        = {1412.5424},
  archivePrefix = {arXiv},
  primaryClass  = {gr-qc}
}

@article{MureikaMoffatFaizal2016,
  author        = {Mureika, Jonas R. and Moffat, John W. and Faizal, Mir},
  title         = {Black Hole Thermodynamics in {MOdified Gravity} ({MOG})},
  journal       = {Physics Letters B},
  volume        = {757},
  pages         = {528--536},
  year          = {2016},
  doi           = {10.1016/j.physletb.2016.04.041},
  eprint        = {1504.08226},
  archivePrefix = {arXiv},
  primaryClass  = {gr-qc}
}

@book{Lee2012,
  author    = {Lee, John M.},
  title     = {Introduction to Smooth Manifolds},
  edition   = {2},
  series    = {Graduate Texts in Mathematics},
  volume    = {218},
  publisher = {Springer},
  address   = {New York, NY},
  year      = {2012},
  doi       = {10.1007/978-1-4419-9982-5},
  isbn      = {978-1-4419-9981-8}
}

@book{DyatlovZworski2019,
  author    = {Dyatlov, Semyon and Zworski, Maciej},
  title     = {Mathematical Theory of Scattering Resonances},
  series    = {Graduate Studies in Mathematics},
  volume    = {200},
  publisher = {American Mathematical Society},
  address   = {Providence, RI},
  year      = {2019},
  doi       = {10.1090/gsm/200},
  isbn      = {978-1-4704-4366-5}
}

@article{WeiLiu2020ThermodynamicQNM,
  author        = {Wei, Shao-Wen and Liu, Yu-Xiao},
  title         = {Null Geodesics, Quasinormal Modes, and Thermodynamic Phase Transition for Charged Black Holes in Asymptotically Flat and {dS} Spacetimes},
  journal       = {Chinese Physics C},
  volume        = {44},
  number        = {11},
  pages         = {115103},
  year          = {2020},
  doi           = {10.1088/1674-1137/abae54},
  eprint        = {1909.11911},
  archivePrefix = {arXiv},
  primaryClass  = {gr-qc}
}

@article{VolkelKonoplyaKokkotas2019Hawking,
  author        = {V{\"o}lkel, Sebastian H. and Konoplya, Roman and Kokkotas, Kostas D.},
  title         = {The Inverse Problem for Hawking Radiation},
  journal       = {Physical Review D},
  volume        = {99},
  number        = {10},
  pages         = {104025},
  year          = {2019},
  doi           = {10.1103/PhysRevD.99.104025},
  eprint        = {1902.07611},
  archivePrefix = {arXiv},
  primaryClass  = {gr-qc}
}

@article{Ovgun2026QuantumOptics,
  author        = {{\"O}vg{\"u}n, Ali},
  title         = {Quasinormal Mode Spectroscopy via Horizon-Brightened Quantum Optics},
  journal       = {The European Physical Journal C},
  volume        = {86},
  pages         = {293},
  year          = {2026},
  doi           = {10.1140/epjc/s10052-026-15543-0},
  eprint        = {2511.07488},
  archivePrefix = {arXiv},
  primaryClass  = {gr-qc}
}

@article{Hatsuda2020Borel,
  author        = {Hatsuda, Yasuyuki},
  title         = {Quasinormal Modes of Black Holes and {Borel} Summation},
  journal       = {Physical Review D},
  volume        = {101},
  number        = {2},
  pages         = {024008},
  year          = {2020},
  doi           = {10.1103/PhysRevD.101.024008},
  eprint        = {1906.07232},
  archivePrefix = {arXiv},
  primaryClass  = {gr-qc}
}

\end{document}